%% file: main.tex
\documentclass[aps,prx,10pt,onecolumn,showpacs,citeautoscript,amsmath,amssymb,floatfix,nofootinbib,superscriptaddress,longbibliography]{revtex4-2}

\input{Preamble.tex}
\begin{document}

\title{Quantum incompatibility of Born probabilities}

\author{Esteban Castro-Ruiz}
    \email{esteban.castroruiz@oeaw.ac.at}
    \affiliation{Institut für Quantenoptik und Quanteninformation (IQOQI), Austrian Academy of Sciences, Boltzmanngasse 3, 1090 Vienna, Austria}
    
\author{Nathan Cohen}
    \email{nathan.cohen@univie.ac.at}
    \affiliation{Institut für Quantenoptik und Quanteninformation (IQOQI), Austrian Academy of Sciences, Boltzmanngasse 3, 1090 Vienna, Austria}
    \affiliation{University of Vienna, Faculty of Physics, Vienna Doctoral School in Physics, and Vienna Center for Quantum Science and Technology (VCQ), Boltzmanngasse 5, A-1090 Vienna, Austria}
\author{Luis C. Barbado}
    \email{luis.cortes.barbado@univie.ac.at}
    \affiliation{University of Vienna, Faculty of Physics, Vienna Doctoral School in Physics, and Vienna Center for Quantum Science and Technology (VCQ), Boltzmanngasse 5, A-1090 Vienna, Austria}
\author{{{\v{C}}}aslav Brukner}
 \email{caslav.brukner@univie.ac.at}
    \affiliation{Institut für Quantenoptik und Quanteninformation (IQOQI), Austrian Academy of Sciences, Boltzmanngasse 3, 1090 Vienna, Austria}
    \affiliation{University of Vienna, Faculty of Physics, Vienna Doctoral School in Physics, and Vienna Center for Quantum Science and Technology (VCQ), Boltzmanngasse 5, A-1090 Vienna, Austria}

\begin{abstract}
\noindent Quantum theory challenges the view that individual measurement outcomes are predefined and independent of the measurement context. Yet the quantum state itself -- the catalogue of probabilities for all possible measurements -- is usually assumed to be well defined. We argue that this assumption tacitly relies on measurements being performed relative to ideal, infinitely-resourceful reference frames. We show that, when measurements are made relative to non-ideal quantum reference frames, the probabilities themselves become indefinite: even in the limit of arbitrarily large number of runs, the relative frequencies may remain uncertain. The uncertainty is irreducible in a quantum-mechanical sense, as we show by proving a Bell-type theorem for relative frequencies. We further propose a quantum-optical implementation of these relational measurements based on pulsed homodyne detection. Our findings motivate an extension of the notion of the quantum state to regimes constrained by finite resources. We expect them to be especially relevant at the interface between quantum theory and general relativity, where the resources and information available in a bounded region of spacetime are fundamentally limited.
\end{abstract}

\maketitle

\tableofcontents

\section{Introduction}
\noindent Quantum theory (QT) is a fundamentally probabilistic theory: rather than predicting individual measurement outcomes, it predicts their probabilities. Unlike classical probabilistic theories, QT challenges the view that these probabilities merely reflect our ignorance of all the variables relevant to the experiment. This challenge is supported by \emph{quantum contextuality}, established for single systems by the Kochen--Specker theorem \cite{KochenSpecker1967} and extended by Bell's theorem \cite{Bell1964} to \emph{quantum nonlocality} for spacelike separated measurements. One possible reading of these theorems is that, in QT, the outcomes of all possible measurements are not jointly well-defined; that is, there exists no joint probability distribution for all of them. Consequently, outcome values are not predefined prior to, and independent of the measurement context. For example, if a Stern--Gerlach (SG) device yields the outcome of a spin measurement along the \( z \)-axis, one cannot simultaneously assign a definite value to the the spin along the \(x \)-axis.

Despite the incompatibility of measurement outcomes, QT predicts well-defined and compatible probability distributions for all possible measurements. In the example above, the probabilities of spin measurements along arbitrary directions are jointly well defined, even though only measurements along the $z$-axis have definite outcomes in the experimental context considered. This mutual compatibility of probability assignments is a central structural feature of QT. Indeed, a quantum state may be viewed as a compact representation of well-defined probabilities for all possible measurements. In this spirit, Schr\"{o}dinger called quantum states an ``\emph{Erwartungskatalog}'' -- a ``\emph{catalogue of expectations}'' \cite{Schrodinger1935}.

Operationally, the connection between how probabilities are computed in quantum mechanics and how they are statistically inferred in real experiments relies on two idealizations. First, one typically assumes that an experiment can be repeated many times under the same conditions. Such repeatability motivates treating the resulting sequence of systems as independent and identically distributed (i.i.d.), and implies that the observed relative frequencies converge, in the limit of arbitrarily many runs, to the probabilities given by the Born rule \cite{Finkelstein1963, Hartle1968, FarhiGoldstoneGutmann1989, Herbut2015prob}. The assumption reflects standard experimental practice, where relative frequencies collected over repeated runs are compared with the theory's predictions, while remaining neutral on the interpretation of probability---frequentist, Bayesian, or otherwise. Second, measurements in QT are usually performed relative to a classical reference frame, which the formalism implicitly assumes to be ideal and therefore does not model it explicitly. Ideal frames define perfectly distinguishable ``orientations'' and are unaffected by the system under consideration, ensuring that the experiment is performed under the same conditions. The second assumption thus underpins the feasibility of the first: in our SG experiment, the magnet is  aligned along a direction defined by the frame, which the frame distinguishes perfectly and remains stable across repeated runs.

However, reference frames are physical systems like any other. In particular, they obey the laws of quantum theory, and one may reasonably expect their quantum properties to become relevant in certain situations. Introducing quantum reference frames (QRFs) lets one describe reference frames explicitly and even allow them to be in genuinely quantum states, such as superpositions or states entangled with the system~\cite{aharonov1967charge,aharonov1984quantum, Giacomini2019,Vanrietvelde_2020,de_la_Hamette_2020,Krumm2021,de2021perspective,carette2025operational, castro2025relative}. Only in the limit where QRFs become ideal and classical-like -- sharply localized in the relevant frame degrees of freedom -- does one recover the standard QT description~\cite{BartlettRudolphSpekkens2007, loveridge2012quantum, loveridge2018symmetry}. Away from this limit, non-ideal frames in genuinely quantum states can give rise to new phenomena~\cite{Bartlett_2006, Poulin_2006,PoulinYard_2007,Miyadera_2016,MikuschEtAl2021_QRFspin,de2021perspective, garmier2025perspectives}. For example, bounded reference frames have been shown to limit the detectability of entanglement and Bell nonlocality~\cite{CostaHarriganRudolphBrukner2009}.

Here we show that, when quantum reference frames are non-ideal, \emph{probabilities in QT can themselves become indefinite and context-dependent}. Constructing relative frequency operators with respect to non-ideal QRFs, we show that the operators associated with two different measurements need not commute, even in the limit of arbitrarily large runs. The corresponding relative frequencies can therefore be incompatible, inheriting the characteristic features of incompatible observables in QT. In particular, they satisfy Heisenberg-type trade-offs: the more sharply the relative frequency of one measurement is defined, the less sharply defined it becomes for an incompatible one. We then consider a Bell scenario in which two observers, Alice and Bob, measure relative frequency operators defined with respect to their non-ideal QRFs. Even in the limit of arbitrarily large number of runs, the correlations between their relative frequencies can violate a Bell inequality. Thus the uncertainty in the relative frequencies accessible in local laboratories cannot be interpreted as ignorance about the ``true probabilities''; rather, it reflects a fundamental impossibility of performing experimental runs under the same conditions, and thus of assigning well-defined probabilities to all conceivable measurements. 
These effects admit a concrete realization in pulsed continuous-variable optics, where the quantum reference frame is a single optical mode and the relational quadrature of an $N$-pulse block is read out by homodyne detection, opening a route to testing the associated Bell inequality.

Our results call into question the standard operational notion of a ``quantum state'' when reference frames are non-ideal, that is, when macroscopic, effectively classical reference frames are unavailable. They suggest that, for each measurement, the probability entering the standard definition of a quantum state may need to be replaced by a probability distribution over possible relative frequencies in the idealized limit of many repetitions---giving rise to a notion of a ``\emph{probability of probabilities}''. The corresponding adapted notion of a quantum state would then encode all these higher-order probability assignments.

These considerations are foundationally motivated by the interplay between quantum theory and general relativity. Indeed, the resources necessary to establish reference frames necessarily lead to gravitational back-action, placing fundamental constraints on the measurability of spacetime quantities~\cite{Padmanabhan1987,NgVanDam1994,Bronstein:1936,PNASpaper-Esteban-Caslav}. Similar motivations have led to the idea that quantum gravity may also require the information-geometric structures underlying the Born rule to be generalized rather than treated as fixed, potentially allowing departures from its standard quadratic form~\cite{MinicTze2003,BerglundEtAl2022Gravitizing,BerglundEtAl2025InformationMetrics}. Relatedly, the framework of Doubly Quantum Mechanics promotes the parameters of symmetry transformations to quantum-mechanical operators, thereby replacing ordinary probabilities by probability operators~\cite{DEsposito2025DoublyQuantum, d2026indefinite}. Finally, given that only finite resources and information are accessible within a bounded spacetime region \cite{Bekenstein1981,Mueller2009FuzzyProbability}, the operational meaning of a quantum state in such regimes is not guaranteed. In this context, our results provide a clear operational framework to investigate these questions, suggesting that the core probabilistic structure of QT may need to be revisited at the interface between our two fundamental theories.

\section{The relative frequency operator}\label{subsec:herbut}
\label{sec:therelativefrequencyoperator}
\noindent Consider an experiment designed to test the predictions of quantum theory. Suppose Alice, an experimenter, wants to verify that the probability of a certain outcome in her experiment is the one given by quantum theory. To do so, she typically runs the experiment on a large number of i.i.d. systems, records the outcomes, and checks whether the observed statistics agree with the theoretical prediction to a given level of confidence.

We now describe this scenario within the formalism of quantum theory. The system of interest, denoted $\mathsf a$, is associated with a Hilbert space $\mathcal H_{\mathsf a}$, and the outcome Alice is interested in is represented by a projection operator $P$ on $\mathcal H_{\mathsf a}$. 
Alice's system is prepared in $N$ i.i.d. copies of a given state $\ket{\psi} \in \mathcal H_{\mathsf a}$, which we write as $\vert \psi_N\rangle:=\vert \psi\rangle^{\otimes N }\in  \mathcal H_{\mathsf a^N}:=\mathcal H_{\mathsf a}^{\otimes N}$. (Here and throughout, we omit the subspace indices whenever there is no risk of confusion.) 
The operator associated with collecting statistics from the $N$ copies is the \emph{relative frequency operator} for the projector $P$, and it is defined as
\be\label{eq:relativefrequencyoperator}
    F_N(P) = \frac{1}{N}\sum_kP^k, \q\q \text{with} \q\q P^k:=\id\otimes\ldots\otimes\id\otimes\overset{\substack{k^{\text{th}} \\ \downarrow}}{P}\otimes\id\otimes\ldots\otimes\id,
\ee
where $\id$ denotes the identity operator on $\mathcal H_{\mathsf a}$. 

To gain some insight into the physical meaning of $F_N(P)$, we write its spectral decomposition as~\cite{Herbut2015prob} 
\begin{equation}
F_N(P) = \sum_{k = 0}^N \Big( \frac{k}{N}\Big) \, Q_{k/N}, \q \q \text{with} \q \q Q_{k/N}=\underbrace{P\otimes\ldots\otimes P}_{k} \otimes \underbrace{P^\perp\otimes ... \otimes P^\perp}_{N-k} + \,  \text{permutations},
\end{equation}
where, $P_\perp = \id - P$. The eigenvalues $\nu = k/N$ are relative frequencies, counting the number of positive outcomes of the quantum test corresponding to $P$ over the total number of trials $N$. 
 For any $N$-copy state $\ket{\psi_N}$, the expectation value of the frequency operator yields
\be
    \langle F_N(P)\rangle_{\psi_N}=\bra{\psi}P\ket{\psi}. 
\ee
Note that \(\vert \psi_N\rangle\) is not, in general, an eigenstate of \(F_N(P)\) for finite \(N\). However, it becomes an approximate eigenstate in the large-\(N\) limit, with asymptotic eigenvalue $\langle\psi\vert P\vert \psi\rangle$. More precisely, the following theorem was proven in Refs.~\cite{Finkelstein1963,Hartle1968,Herbut2015prob}: 

\begin{theorem}[Finkelstein-Hartle-Herbut]\label{theorem: Finkelstein-Hartle-Herbut}
    Given an $N$-copy state $\vert \psi_N\rangle=\vert \psi\rangle^{\otimes N}$ and a projector $P$, the mean value of the frequency operator is $\langle\psi_N\vert F_N(P)\vert \psi_N\rangle=p$ where $p=\langle\psi\vert P\vert \psi\rangle$, and the variance of the frequency operator will vanish for large enough N:
    \be
        \lim_{N\rightarrow\infty} \ls \Big( F_N(P) - p \Big) \vert \psi_N\rangle \rs\longrightarrow 0.
    \ee
\end{theorem}
 As an important corollary, we have that the commutator between the relative frequency operators corresponding to different, possibly non-commuting projectors $P_1$ and $P_2$ vanishes in the large-$N$ limit:
    \be \label{eq:corollary}
        \lim_{N\rightarrow\infty} \ls \Big[F_N(P_1),F_N(P_2) \Big]\vert \psi_N\rangle\rs\longrightarrow 0.
    \ee
Theorem~\ref{theorem: Finkelstein-Hartle-Herbut} shows that, for large $N$, any $N$-copy state $\ket{\psi_N} = \ket{\psi}^{\otimes N}$ approaches an eigenstate of the frequency operator with eigenvalue $\bra \psi P \ket \psi$ -- the textbook probability obtained by applying the Born rule to a single system $\mathsf a$.
Because relative frequencies converge to $\bra \psi P \ket \psi$, we refer to their $N \to \infty$ limit as the \emph{Born probabilities}. By Eq.~\eqref{eq:corollary}, the Born probabilities corresponding to the possibly non-commuting projectors $P_1$ and $P_2$ are jointly well defined.
Theorem \ref{theorem: Finkelstein-Hartle-Herbut} extends naturally to the expectation values of an observable (Hermitian operator). Consider an observable $A$ on $\mathcal H_{\mathsf a}$ with spectral decomposition $A = \sum_ia_i P_i$.  
Its expectation value is $\langle A\rangle=\sum_i a_ip_i$, where $p_i=\bra{\psi}P_i\ket{\psi}$ are probabilities. Using the linearity of the relative frequency operator, we define the \emph{expectation value operator} $F_N(A) = \sum_i a_i F_N(P_i) = (1/N) \sum_k A^k $. It then follows that
\( \langle F_N(A)\rangle_{\psi_N}=\sum_i a_i\langle F_N(P_i)\rangle_{\psi_N} = \sum_i a_i\bra{\psi}P_i\ket{\psi} \),
and
\be
    \lim_{N\rightarrow\infty} \ls \Big( F_N(A) - \langle A \rangle \Big)\vert \psi_N\rangle \rs\ \longrightarrow 0.
\ee
As in Eq.~\eqref{eq:corollary}, the commutator of the relative expectation value operators associated with two (possibly non-commuting) observables $A_1$ and $A_2$ vanishes in the large $N$ limit:
\begin{equation}\label{eq:revcommutativity}
        \lim_{N\rightarrow\infty} \ls \Big[F_N(A_1),F_N(A_2) \Big]\vert \psi_N\rangle\rs\longrightarrow 0.
\end{equation}
We close this section with three remarks.
First, in no practical experimental situation is it meaningful to take the mathematical limit 
$N\to\infty$. Rather, the limit should be understood as referring to a finite but arbitrarily large $N$, for which the relevant probabilities have stabilised to the required accuracy. This is what we mean by the limit in this paper. 
Second, as mentioned above, we interpret Theorem~\ref{theorem: Finkelstein-Hartle-Herbut} in the operational sense relevant to everyday experimental practice in quantum laboratories, where large statistics of outcome sequences are collected and compared with theoretical predictions. Its significance for the interpretation of probability itself has been discussed elsewhere; see, for example, Refs.~\cite{FarhiGoldstoneGutmann1989,CavesSchack2005FrequencyOperator}. Finally, we assumed the $N$-copy state to be i.i.d.,
$\ket{\psi_N} = \ket{\psi}^{\otimes N}$. This state assignment corresponds to the idealized situation in which the preparation device operates repeatedly ``under the same conditions''. In realistic scenarios, the i.i.d. assumption hardly holds exactly. In particular, Bell experiments with
correlations across different pairs provide an important example in which the
usual i.i.d. structure of repeated trials can fail in a physically relevant
way~\cite{AshhabMaruyamaBruknerNori2009,GallegoDakic_2021}. However, an operational route to the notion of a quantum state, accounting for systems that are not necessarily i.i.d., is provided by the (quantum) de Finetti theorem, which can be explained as follows. 

From a Bayesian perspective, before collecting data one assigns a joint state $\rho^{( N)}$ to the $ N$ systems produced by the device, and considers sequences of states for increasing $N$. If the sequence of assigned states is \emph{exchangeable}, meaning that for every $N$ the state $\rho^{(N)}$ is invariant under arbitrary permutations of the $N$ subsystems, and \emph{extendible}, meaning that for every $N$ there exists, for any $M>0$, a state $\rho^{(N+M)}$ on $N+M$ systems that is itself permutation-invariant and whose partial trace over the last $M$ systems equals $\rho^{(N)}$, then the quantum de Finetti representation theorem implies that, for every finite $N$, the joint state admits the form~\cite{Renner2007,ChristandlKoenigMitchisonRenner2007}
\begin{equation}\label{eq:de Finetti}
    \rho^{(N)} = \int \rd\sigma\, \mu(\sigma)\,\sigma^{\otimes N},
\end{equation}
for some probability measure $\mu(\sigma)$ over single-system density operators $\sigma$~\cite{HudsonMoody1976,CavesFuchsSchack2002,Renner2007}. Operationally, this means that an experimenter’s predictions for any finite set of runs are equivalent to assuming that each system is an independently prepared state $\sigma$, drawn from a prior distribution $\mu(\sigma)$.

\section{The relational relative frequency operator}\label{rrfs}

\noindent In describing Alice's experiment in the previous section, we implicitly assumed that she has access to an ideal reference frame with respect to which she measures the system. In this section we make the frame explicit. Moreover, we consider the case in which the frame is non-ideal and may be given a quantum description -- a non-ideal, quantum reference frame (QRF).
Suppose Alice performs measurements on $N$ copies of a system $\mathsf a$ relative to her (possibly \emph{non-ideal}) QRF $\mathsf A$.
By ``non-ideal'' we mean that $\mathsf A$ lacks the resources to define an arbitrarily precise localization of $\mathsf a$ in configuration space. We assume $\mathsf A$ and $\mathsf a$ are described by the Hilbert spaces $\mathcal H_\mathsf{A}$ and $\mathcal H_\mathsf{a}$, carrying representations $U_{\mathsf A}(g)$ and $U_{\mathsf a}(g)$ of the group $G$, respectively. As Alice performs measurements on $N$ copies of $\mathsf a$, the Hilbert space of the experiment is $\mathcal H = \mathcal H_\mathsf{A} \otimes \mathcal H_{\mathsf{a}^N}$.

We assume Alice has no access to any (possibly fundamentally inexistent) external reference frame. Therefore, all she can access are \textit{$G$-invariant quantities}, represented by operators $T$ on $\mathcal H_\mathsf{A} \otimes \mathcal H_{\mathsf{a}^N}$ such that $[U_\mathsf{A}(g)\otimes U_{\mathsf{a}^N}(g), T] =0$ for all $g \in G$. In particular, she can measure \textit{relational observables}, obtained by ``relativising'' observables on $\mathcal H_{\mathsf{a}^N}$ with respect to her reference frame. That is, to any observable $A$ on $\mathcal H_{\mathsf a^N}$ we assign a relational observable on $\mathcal H_{\sf A}\otimes H_{\mathsf a^N}$ according to the prescription \cite{kitaev2004superselection,BartlettRudolphSpekkens2007, loveridge2012quantum, loveridge2018symmetry}
\begin{equation}
A \to \widetilde A = \mathcal G[E(e) \otimes A]. 
\end{equation}
Here $\mathcal G$ denotes $G$-twirling, 
\begin{equation}
\mathcal G [T] = \int_G \mathrm d g \; U \otimes U \; T \;  U^\dagger  \otimes U^\dagger,
\end{equation}
for all operators $T$ on $\mathcal H_{\sf A}\otimes H_{\mathsf a^N}$, and $E(e)$ is the ``seed element'', corresponding to the identity $e \in G$, of a covariant POVM (positive operator-valued measure). By definition a covariant POVM is formed by operators $0 \leq E(g):= g \triangleright E(e) \leq \id$ such that $\int_G \rd g \, E(g) = \id$, where $\rd g$ is the normalised Haar measure on $G$. Throughout, we use the notation $g \triangleright T$ to denote the action of $g\in G$ on a generic object $T$; for instance, if $T$ is an operator on a Hilbert space carrying the representation $U_g$, then  $g  \triangleright T := U_g T U^\dagger_g$. The operator $\widetilde A$ has the physical interpretation of ``measuring $A$ relative to the QRF $\mathsf A$''. 

In this work, we are interested in the relative frequencies measured relative to $\mathsf A$, so we introduce the \textit{relational relative frequency} (RRF) operator,
\begin{equation}\label{eq:rrf-pvm-construction}
    \widetilde{F}_N(P) := \cG\l[ E(e) \otimes F_N(P)\r] = \int_G\rd g \, E(g) \otimes g \triangleright F_N(P) = \int_G\rd g \, E(g) \otimes F_N(g\triangleright P).
\end{equation}
For the rest of the paper, we assume POVMs defined by coherent states $\ket{g}$, setting
$E(g)= (1/r)\vert g\rangle\langle g\vert $ with $(1/r)\int \rd g \, \ketbra{g}{g} = \id$. The inner product
$\langle g\vert g'\rangle \neq\delta(gg'^{-1})$, for $g,g'\in G$, characterises the non-ideal nature of the QRF. Inserting an observable $A$ in place of a projector in Eq.~\eqref{eq:rrf-pvm-construction}, we may also define a \emph{relational expectation value} (REV) operator $\widetilde F_N(A)$ in a completely analogous way.

Crucially, we construct Eq.~\eqref{eq:rrf-pvm-construction} by first forming the relative frequency operator $F_N(P)$ and then making a relational observable out of it. We therefore have a \emph{single copy} of the QRF $\mathsf A$ per $N$ copies of the system $\mathsf a$. Had we proceeded the other way around -- first forming a relational operator for a single copy of $\mathsf a$ and then building the corresponding relative frequency operator -- we would have ended up with $N$ independent copies of $\mathsf A$, requiring far more resources than Alice has in her lab. 

When thinking operationally about relational relative frequencies in the absence of an external reference frame, Alice may implement two different experimental procedures, each related to the RRF operator~\eqref{eq:rrf-pvm-construction}. In the first, which we call the POVM construction, Alice measures the POVM obtained by \(G\)-twirling the spectral projectors of the non-relativised relative frequency operator \(F_N(P)\). Here \(\widetilde F_N(P)\) is the first-moment operator of the resulting POVM, with weights given by the eigenvalues of \(F_N(P)\). Concretely, take the spectral decomposition of $F_N(P)$ presented in Section \ref{sec:therelativefrequencyoperator},
   $ F_N(P)=\sum_\nu \nu\, Q_\nu$,
    where $\nu =k/N$
are the possible relative frequencies and
\(Q_\nu\) are the corresponding spectral projectors. Defining $\widetilde{Q}_\nu
    =
   (1/r)\mathcal G[\ketbra{e}{e}\otimes Q_\nu]$, we obtain
\begin{equation}
   \widetilde F_N(P)= \sum_\nu \nu\,\widetilde{Q}_\nu. 
    \label{eq:rrf-povm-first-moment}
\end{equation}
The operators
\(\{\widetilde{Q}_\nu\}_\nu\) form a (generally non-projective) POVM whose first moment is exactly the RRF operator. In the second, which we call the PVM (projective-valued measure) construction, the Hermitian operator \(\widetilde F_N(P)\) defines a projective measurement whose outcome distribution is determined by its spectral projectors:
\begin{equation}
    \widetilde F_N(P)=\sum_\lambda \lambda\, \Pi_\lambda ,
    \label{eq:rrf-pvm-spectral-decomposition}
\end{equation}
where the eigenvalues \(\lambda\) are real and the \(\Pi_\lambda\) are
mutually orthogonal projectors. 

The  POVM and PVM constructions coincide in the limit of ideal reference frames, and for sharp, ideal QRF states $\ket{g}$, they are equivalent to the projective measurement defined by the relative frequency operator~\eqref{eq:relativefrequencyoperator},
thereby reproducing the predictions of ordinary quantum theory. 
For non-ideal QRFs, however, the constructions are generally distinct. 

In the POVM construction, the outcomes \(\nu\) are the same relative frequency values \(k/N\) as for the non-relativised frequency operator. In the scenario considered here, which assumes only a single copy of the QRF per $N$ system copies, however, their natural operational interpretation requires an ideal external reference frame for the relevant transformation group. As discussed in Ref.~\cite{garmier2025perspectives}, a general relational POVM may be implemented by first measuring the orientation of the QRF relative to the external frame and subsequently, conditioned on the result, performing the correspondingly transformed measurement on the system. The QRF and the system are thus measured separately relative to the external frame, and their relational statistics are obtained by combining the two measurement records. In Appendix~\ref{app:povmhiddenvariable} we support this interpretation by proving that the reduced POVM on the system is equivalent to a classical mixture of measurements on the system defined relative to the ideal external reference frame, consistent with Ref.~\cite{loveridge2018symmetry}.

By contrast, the PVM construction cannot be reduced to such a classical average, and therefore does not require an external reference frame to be implemented in the scenario we study. It is for this reason the construction of primary interest here.
To implement it, Alice couples the systems
\(\mathsf a_1,\ldots,\mathsf a_N\) and measures the resulting collective system relative to her QRF \(\mathsf A\). This is consistent with our assumption that she holds a single copy of \(\mathsf A\) for each block of \(N\) copies of \(\mathsf a\) and has no access to an external reference frame: in a single shot, she measures the entire \(N\)-system block collectively. This procedure is somewhat analogous to the measurement of mean magnetization, from which one can infer the relative frequencies. However, according to Eq.~\eqref{eq:rrf-pvm-spectral-decomposition} the relational relative frequencies of the PVM construction need not coincide with the rational numbers $k/N$ that label ordinary finite-$N$ relative frequencies. For non-ideal QRFs, relativisation turns rational-valued relative frequency operators into noncommuting operator sums. Hence their spectra need not remain rational.
In Appendix~\ref{URPOVM}, we show that the PVM is the sharpest possible measurement of $\widetilde F_N(P)$.

Although we adopt the PVM construction as the operational definition of measuring $\widetilde F_N(P)$ (for the reason given above), we also study the uncertainty relations stemming from the POVM construction in Subsection~\ref{subsec:heisenbergtradeoff}, and show in Appendix~\ref{app:povmhiddenvariable} that this construction cannot violate the Bell inequality. 

We now present two concrete examples of RRF operators which will be useful later on.

\subsection{RRF operator for SU(2)} \label{subsec:rrfforsu2}
In order for the reference frame states to resolve the full $\SU(2)$ group, we consider a ``truncated regular representation'' as the QRF (see e.g.~\cite{de2021perspective,garmier2025perspectives}). The Peter-Weyl theorem states that the regular representation decomposes into sectors of well-defined total angular momentum $J$, and that each $J$-sector is spanned by (not normalised) vectors of the form
\begin{equation}\label{eq:regularrepresentation}
    \ket{g} = \sqrt{d_J}\sum_{m,n}D^J_{mn}(g)\,\ket{J,m}_L\otimes\ket{J,n}^*_R,
\end{equation}
where $D^J(g)$ is the Wigner D-matrix of spin $J$, $d_J=2J+1$, and $L$ and $R$ denote respectively the left-regular and right-regular tensor factors of the reference frame (see Appendix~\ref{app:uncertaintyrelations} for details). We adopt the convention in which the action of $\SU(2)$ is given by the left-regular representation, while the right-regular representation acts on the ($\SU(2)$-invariant) multiplicity space.

For our purposes, it suffices to assume that $\mathsf A$ lives in a single $J$-sector.
With the normalisation~\eqref{eq:regularrepresentation}, it holds that \(\int  \rd g\ketbra{g}{g}  = \id_{\mathsf A}\), and we obtain (see Appendix~\ref{app:uncertaintyrelations})
\begin{equation}\label{eq:niceformofsu2relativefrequencies}
     \widetilde{F}_N(P_{\mathbf{v}})= \mathcal G [\ketbra{e}{e}\otimes  F(P_{\mathbf v})]= \frac{1}{2} \id_{\mathsf A \mathsf a^N} + \frac{1}{N} \frac{1}{J(J+1)} \, (\mathbf{v}\cdot \mathbf{R}_{\mathsf A}) \otimes{\mathbf{L}_{\mathsf A}} \cdot {\mathbf{S}}_{\mathsf a^N},
\end{equation}
where $\mathbf{R_{\mathsf A}} = (\id \otimes R_x,\id \otimes R_y,\id \otimes R_z)$ is the right-regular angular momentum vector operator on $\mathsf A$,  $\mathbf{L}_{\mathsf A} = (L_x \otimes \id, L_y  \otimes \id, L_z \otimes \id)$ is the left-regular angular momentum operator on $\mathsf A$ and $\mathbf{S}_{\mathsf a^N}$ is the angular momentum vector operator on the $N$ copies of $\mathsf a$. Note that the unit vector $\mathbf v$, specifying the direction in which spin is measured via the projector $P_{\mathbf v}=\frac{1}{2}(\id + \mathbf v \cdot \mathbf \sigma)$, couples to the multiplicity degrees of freedom of the right-regular representation, whereas the non-invariant left regular angular momentum $\mathbf L_{\mathsf A}$ couples to the (non-invariant) angular momentum vector on $\mathsf a^N$ via $\mathbf{L}_{\mathsf A}\cdot \mathbf{S}_{\mathsf a^N}$. Intuitively, this coupling means that Alice measures ``how many'' spins of the system $\mathsf a$ are ``aligned'' with her QRF $\mathsf A$.

The spectrum of $\widetilde F_N(P_{\mathbf v})$ lies inside $\l[\frac{1}{2}\frac{1}{J+1},\frac{1}{2} + \frac{1}{2}\frac{J}{J+1}\r]$. In particular, for $J=1/2$, the probabilities lie in $[\frac{1}{3}, \frac{2}{3}]$. As in the non-relational case, the action of \(\widetilde{F}_N(P_{\mathbf v})\) on a product state
\(\ket{\chi}=\ket{e}\otimes\ket{\psi}^{\otimes N}\)
still contains \(N\)\  and $J$-dependent off-diagonal terms, which vanish only in the limit \(N,J\to\infty\).

\subsection{REV operator for the Heisenberg-Weyl group}
\label{subsec:heisenbergweyl}
 Consider now the case where both the QRF $\mathsf A$ and the system $\mathsf a$ are particles on one-dimensional line, whose Hilbert spaces carry representations of the Heisenberg-Weyl group, generated by the dimensionless position and momentum operators $ X = (a + a^\dagger)/\sqrt 2$ and $P = i(a - a^\dagger)/\sqrt 2$. We are interested in studying REV operators for arbitrary quadratures  $X_\vartheta\;=\;\cos\vartheta\, X+\sin\vartheta\, P$.

By the linearity of the REV operators $\widetilde F_N(X)$ and $\widetilde F_N(P)$, we can define the REV operator of an arbitrary quadrature $X_\vartheta$. In Appendix~\ref{app:uncertaintyrelations} we show that it may be written as
\begin{equation}
\widetilde F_N( X_\vartheta) \; = \; \frac{1}{\pi}\mathcal G[\ketbra{\alpha =0}{\alpha =0} \otimes F_N(X_\vartheta)] \; = \; \id_\mathsf{A}\otimes F_N( X_\vartheta)\;-\; X_\vartheta\otimes\id_{\mathsf{a}^N},
\label{eq:relquadmain}
\end{equation}
where $\ket{\alpha=0}$ is the ``vacuum'' coherent state. Let $\ket{q}_\vartheta$ denote the eigenvector of $X_\vartheta$ with eigenvalue $q$, so that
$\bra{q}q'\rangle_{\vartheta}=\delta(q-q')$, and $\int_{\mathbb R} \mathrm d q\,\ket{q}\!\bra{q}_\vartheta = \id$.
On $\mathsf{a}^N$, define
$\ket{\vec q \,}_\vartheta\;:=\;\bigotimes_k\ket{q_k}_\vartheta \in \mathcal H_{\mathsf a^N}$, where
$\vec q \,=(q_1,\dots,q_N)\in\mathbb R^N$ and $k = 1,\dots,N$. These vectors satisfy
$ X_{\vartheta,k}\ket{\vec q \,}_\vartheta=q_k\ket{\vec q \,}_\vartheta$, and
$F_N(X_\vartheta)\ket{\vec q \,}_\vartheta= \bar q \ket{\vec q \,}_\vartheta$,
where $\bar q := (1/N)\sum_{k=1}^{N}q_k$. Both terms in Eq.~\eqref{eq:relquadmain} are then diagonal on the product state $\ket{q}_\vartheta\otimes\ket{\vec q \,}_\vartheta$, giving
\begin{equation}
\widetilde F_N( X_\vartheta)\,\bigl(\ket{q}_\vartheta\otimes\ket{\vec q \,}_\vartheta\bigr)\;=\;
(\bar q-q)\,\ket{q}_\vartheta\otimes\ket{\vec q \,}_\vartheta .
\label{eq:eigenvalue}
\end{equation}
The PVM Alice performs therefore amounts to measuring the mean quadrature (in direction $\vartheta$) of the joint system $\mathsf{a}^N$ relative to the single copy of $\mathsf A$. This differs from the mean relative quadrature of a single copy, averaged over $\mathsf{a}^N$, which would require multiple copies of the QRF.

\section{Uncertainty relations}\label{sec:uncertaintyrelations}

\noindent In ordinary quantum theory, the relative frequencies of any two effects, as well as the expectation values of any two different observables, are always jointly well defined for any state. This holds even when the effects are incompatible or the observables do not commute, and may be traced to the commutativity of the relative frequency and expectation value operators in Eqs.~\eqref{eq:corollary} and~\eqref{eq:revcommutativity}. In assigning definite probabilities and expectation values, ordinary quantum theory tacitly assumes an ideal reference frame. Making this assumption explicit, we show in this section that, when the (quantum) reference frame is non-ideal, the relative frequencies and expectation values of two observables may fail to be jointly well defined, even in the limit of arbitrarily many system copies. Specifically, the corresponding relational relative frequency and relational expectation value operators do not generally commute as $N \to \infty$, leading to Heisenberg-type uncertainty relations for relative frequencies and expectation values.

\subsection{Non-commutativity}
\label{sec:uncertainty}
We first show that, when Alice has access only to a single, non-ideal QRF $\mathsf A$, measurements of the relative frequencies of different effects with respect to $\mathsf A$ are \textit{not necessarily} jointly definable -- not even in the limit of infinitely many runs. The core result is the following:

\begin{theorem}
\label{proposition: non-commutativity}
    Let $\vert \phi\rangle_{\mathsf A}$ be the (pure) state assigned to Alice's quantum reference frame $\mathsf A$, $\vert \psi_N\rangle = \ket{\psi}^{\otimes N} \in\cH_{\mathsf a}^{\otimes N}$ be an $N$-copy state of the system $\mathsf a$, $P_1$ and $P_2$ be projectors on $\cH_{\mathsf a}$ and $E(g)$ the covariant POVM element on $\mathcal H_\mathsf A$ associated to the group element $g$.
    Then, the relational relative frequency operators $\widetilde F_N(P_1)$ and $\widetilde F_N(P_2)$ acting on $\ket{\chi}=\ket{\phi}\otimes \ket{\psi_N}$ satisfy
    \begin{equation}\label{eq: non-commutativity}
\lim_{N\rightarrow\infty}\ls\Big[\widetilde{F}_N(P_1),\widetilde{F}_N(P_2)\Big]\vert \chi\rangle\rs = \ls \Big[M_{P_1},M_{P_2} \Big]\vert \phi\rangle\rs,
    \end{equation}
    with 
    \begin{equation}
            M_{P_1}  := \int_G\rd g \,  \langle\psi\vert g\triangleright P_1\vert \psi\rangle E(g) \qquad \text{{and}} \qquad 
            M_{P_2}  := \int_G\rd g \, \langle\psi\vert g\triangleright P_2\vert \psi\rangle E(g). 
    \end{equation}
\end{theorem}
\begin{proof}
    Refactorising the commutator gives
    \begin{equation}\label{eq:refactorising}
        \begin{aligned}
            \Big[ \widetilde{F}_N(P_1),\widetilde{F}_N(P_2) \Big] = \frac{1}{2}\int_G \rd g\rd h \Big([E(g),E(h)]   \otimes & \{ F_N(g\triangleright P_1),F_N(h\triangleright P_2) \} \\  
            +  \{E(g),E(h)\}\otimes & [F_N(g\triangleright P_1),F_N(h\triangleright P_2)]\Big), 
        \end{aligned}
    \end{equation}
where $\{E(g),E(h)\}$ denotes the anticommutator of $E(g)$ and $E(h)$. Note that Theorem \ref{theorem: Finkelstein-Hartle-Herbut}, as well as its corollary in Eq.~\eqref{eq:corollary}, holds for all individual $g$'s 
    \begin{equation}\label{eq:herbutforeachg}
        \lim_{N\rightarrow\infty} \ls \Big( F_N(g\triangleright P_1)-\langle\psi\vert g\triangleright P_1\vert \psi\rangle \Big)\vert \psi_N\rangle\rs= 0.
    \end{equation}
 Plugging Eq.~\eqref{eq:herbutforeachg} into Eq.~\eqref{eq:refactorising} applied to $\ket{\chi}$ yields Eq.~\eqref{eq: non-commutativity}.
\end{proof}
\noindent A few remarks are in order. Most importantly, because Alice's QRF is non-ideal, the right hand side of Eq.~\eqref{eq: non-commutativity} does not vanish in general. However, it does vanish in the limit where $\langle g\vert g'\rangle_R=\delta(gg'^{-1})$, i.e. when the frame is ideal, recovering Theorem~\ref{theorem: Finkelstein-Hartle-Herbut} and its corollary. Second, the projectors $P_1$ and $P_2$ need not be non-commuting for having non-commuting relative frequency operators. Finally, an analogous result clearly holds for the commutator of two REV operators, some consequences of which we explore in Subsection~\ref{subsec:uncertaintyofexpectation}.

Theorem~\ref{proposition: non-commutativity}
implies that, when Alice's QRF is non-ideal, the $N\to\infty$ limit of relative frequencies, namely the Born probabilities, may remain uncertain. Indeed, a non-vanishing Eq.~\eqref{eq:refactorising} means that the probability distributions of two noncommuting RRFs cannot both be sharp. Figure~\ref{fig:convergence_finite_J} illustrates the operational situation Alice may encounter when measuring the relative frequency of a given outcome using an ideal (upper panel) or non-ideal (lower panel) reference frame. For an ideal QRF we recover the standard case, in which the relative frequency converges to a single probability value in accordance with the law of large numbers. For every finite $N$, the estimate is associated with a confidence interval that scales as $ \sim 1/\sqrt{N}$. 
By contrast, when the QRF is non-ideal, the relative frequency operator is not sharp, so the relative frequencies obtained from different series of \(N\) repetitions need not converge to the same probability value, even in the idealized limit of infinitely many runs. While the confidence interval initially narrows as $N$ grows, this behaviour does not persist as $N$ increases further; instead, it saturates at a constant width.

\begin{figure*}[htbp]
    \centering
    \includegraphics[width=0.7\linewidth]{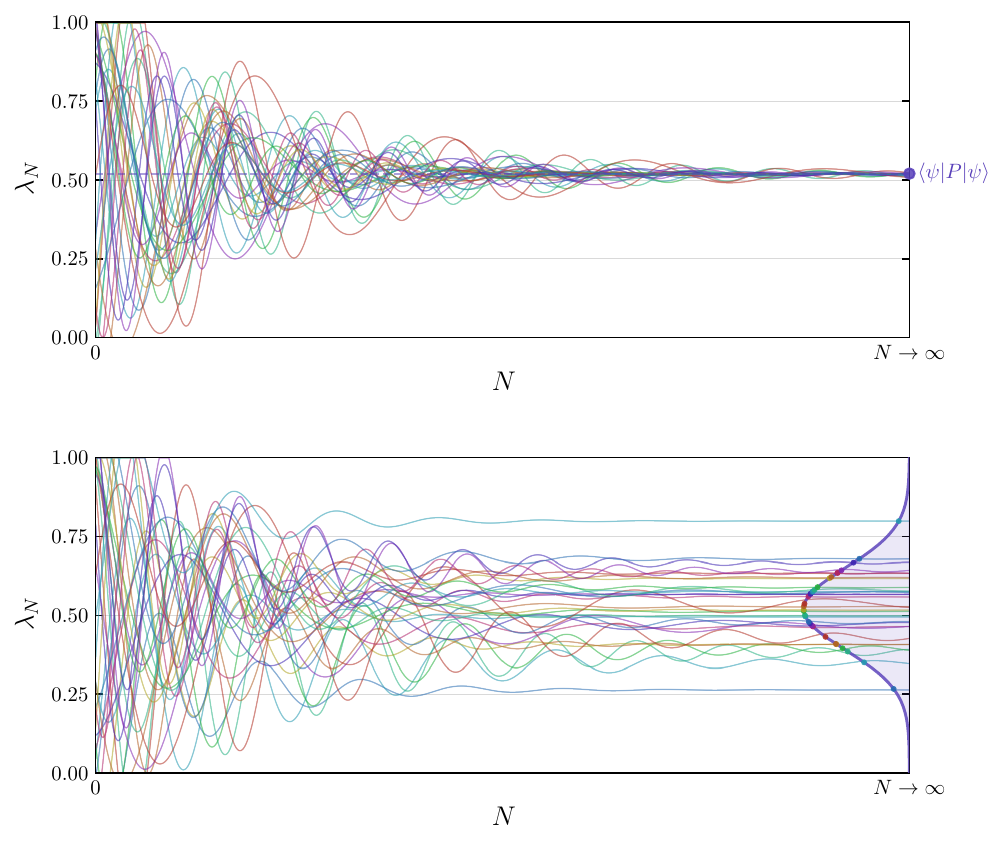}
    \caption{\label{fig:convergence_finite_J} \justifying In ordinary quantum theory, the value of an observable may be uncertain, but the probabilities are sharply defined. The Born rule assigns each outcome a fixed number $p = \bra \psi  P\ket \psi$, and the law of large numbers guarantees that the relative frequencies recorded over $N$ system copies converge to these probabilities as $N\to\infty$: individual outcomes may be uncertain, but the statistics governing them are not. This, however, presupposes an ideal reference frame with respect to which the measurements are done. When Alice's QRF $\mathsf{A}$ is non-ideal, this is no longer the case. The $N\to\infty$ limit of the relative frequencies --- namely, the Born probabilities --- remain uncertain. This figure is a qualitative illustration of the effect. Here, each coloured curve shows how a single relational relative frequency $\lambda_N$ may vary from run to run as the number $N$ of systems increases.
    \\ \textbf{Upper panel:} Relative to an ideal quantum reference frame, the relative frequencies are given by the relational relative frequency operator \(\widetilde F_N(P)\) on a state \(\ket{\chi}\), featuring an i.i.d. preparation of the system. As the number of copies grows arbitrarily, the relative frequencies converge and coincide with the single-system probabilities as given by the Born rule (see Theorem~\ref{theorem: Finkelstein-Hartle-Herbut}). In this case, the lines can be understood as describing the change of relative frequencies in a single experiment as the number of systems increases. \\ \textbf{Lower panel:} When measured with respect to a non-ideal quantum reference frame, the relational relative frequencies corresponding to the operator $ \widetilde F_N(P) $ and the state $\ket \chi$ remains uncertain even in the limit of infinitely many system copies (see Theorem~\ref{proposition: non-commutativity}). In this case, the lines should only be understood figuratively, as our computation involves a collective measurement of $N$ copies relative to the non-ideal frame. While run-to-run changes of the relational relative frequency lie outside our framework, the mean and the variance of the distribution of RRFs can be computed within it.
    In some cases, the spectrum of the RRF operator may be a strict subset of $[0,1]$.}
\end{figure*}

In the following, we derive Heisenberg-like uncertainty relations for RRF and REV operators in the large-$N$ limit in both the PVM and POVM constructions, followed by their concrete examples for the \(\SU(2)\), Heisenberg--Weyl, and centrally extended Galilei groups. 

\subsection{Uncertainty relations} \label{subsec:heisenbergtradeoff}
Ordinary quantum theory tacitly assumes the use of an ideal reference frame. The expectation values \(\langle  A_1 \rangle_\chi\) and \(\langle A_2 \rangle_\chi\) of any two observables \( A_1 \) and \( A_2\) are therefore jointly well defined for any state \(\ket{\chi}\), even when \( A_1\) and \(A_2\) do not commute. What non-commutativity affects is not the expectation values themselves but the variances of the observables, as expressed by the Heisenberg uncertainty relation. By contrast, in the presence of non-ideal reference frames, RRFs and REVs cannot in general, be simultaneously sharp: they already satisfy an uncertainty relation, a direct consequence of the non-commutativity of RRFs (Eq.~\eqref{eq: non-commutativity}) and its analogous result for REVs. In the following, we derive concrete instances of this Heisenberg-like trade-off.

Consider a Hermitian operator $A$ with spectral decomposition $A = \sum_n a_n P_{n}$ and its corresponding REV operator $\widetilde F_N(A)$. In the case where all $a_n$s are zero except for one, $A$ is a projector, $A = P$, and is associated with the RRF operator $\widetilde F_N(P)$. In either case, the results of the previous subsection imply the following Heisenberg-like trade-off for any state $\ket{\chi}$: 

\begin{equation}
    \Delta_\chi^2(\widetilde F_N( A_1))\Delta_\chi^2(\widetilde F_N( A_2)) \geq \frac{1}{4}\Big\vert \Big\langle \Big[\widetilde F_N( A_1), \widetilde F_N( A_2) \Big]\Big\rangle_\chi\Big\vert^2. \label{eq:uncertaintymain}
\end{equation}
For ideal QRFs, $\Delta^2_\chi(\widetilde F_N( A))$, vanishes as $N \rightarrow \infty$ for all states of the form $\ket{\chi} = \ket{\Phi} \otimes \ket{\Psi}^{\otimes N} \in \mathcal{H}_\mathsf{A} \otimes \mathcal H_{\mathsf{a}^N}$, where $\ket{\Phi}$ is an arbitrary state of an ideal QRF $\mathsf{A}$ and $\ket{\Psi}^{\otimes N}$ is an i.i.d. state of $N$ copies of systems $\mathsf a$. 

While uncertainty relation of Eq.~\eqref{eq:uncertaintymain} applies to the PVM construction,
an analogous relation holds for the POVM construction. For concreteness, take $A_1 = P_1$ and $A_2 = P_2$ to be projectors -- the RRF case -- though our results extend to the REV case as well. Recall that, in the POVM construction, the probability distribution over the outcomes of RRF measurements is given by the operators $\widetilde Q_\nu$ of Eq.~\eqref{eq:rrf-povm-first-moment}, of which the RRF operator $\widetilde F_N(P) = \sum_\nu \nu \,  \widetilde Q_\nu$ is the first moment. The second moment, \(\widetilde F^{(2)}_{N,\mathrm{POVM}}(P):=
\sum_{\nu}\nu^2\widetilde Q_{\nu}\), is in general not equal to  $ (\widetilde F_N(P))^2$, because the $\widetilde Q_\nu$s are not projectors. Let \(\Delta^2_{\mathrm{POVM},\chi}(\widetilde F_N(P)) := \langle \widetilde F^{(2)}_{N,\mathrm{POVM}}(P) \rangle_\chi - \langle \widetilde F_N(P)\rangle_\chi^2 \) denote the variance of the observed POVM outcome distribution for a single projector \(P\) in state \(\ket{\chi}\). As shown in Appendix~\ref{URPOVM}, and in agreement with Ref.~\cite{massar2007uncertainty}, this
variance decomposes as
\begin{equation}
    \Delta^2_{\mathrm{POVM},\chi}(\widetilde F_N(P))
    =
    \Delta^2_{\chi}(\widetilde F_N(P))
    + \langle \mathcal N_P \rangle_\chi
\end{equation}
where $\Delta^2_{\chi}(\widetilde F_N(P))$ is the variance of the PVM construction and  \(\mathcal N_P\geq0\) is an intrinsic ``noise operator'' of the POVM. Thus, already for a single projector $P$, the POVM construction can only increase the variance relative to that obtained in the PVM construction.

Applying this observation to two projectors \(P_1\) and \(P_2\), and
applying Eq.~\eqref{eq:uncertaintymain} to \(\widetilde F_N(P_1)\) and \(\widetilde F_N(P_2)\), we obtain
\begin{equation}
    \Delta^2_{\mathrm{POVM},\chi}(\widetilde F_N(P_1))\,
    \Delta^2_{\mathrm{POVM},\chi}(\widetilde F_N(P_2))
    \geq
    \frac14
    \left\vert 
        \Big\langle      
            \Big[\widetilde F_N(P_1),\widetilde F_N(P_2) \Big] \Big\rangle_\chi        
    \right\vert ^2 .
\end{equation}
Therefore the POVM construction inherits the incompatibility of the relational
relative frequency first moments. The difference from the PVM construction is
that the actually measured POVM variances contain additional positive noise
contributions.

We now study three concrete cases of uncertainty relations for RRFs and REVs in the PVM construction.

\subsubsection{Uncertainties for relational relative frequencies}

Suppose Alice measures the RRF observables $\widetilde F_N(P_{\mathbf v})$ and $\widetilde F_N(P_{\mathbf w})$, corresponding to spin measurements along the directions $\mathbf v$ and $\mathbf w$ (see Eq.~\eqref{eq:niceformofsu2relativefrequencies}). Assume that her QRF is in the state $\ket{\phi}$ and the $N$ systems $\mathsf a$ are prepared in the i.i.d. state $\ket{\psi_N} = \ket{\mathbf u, +}^{\otimes N}$, where $\mathbf{u}\cdot \mathbf{\sigma} \ket{\mathbf u, +} = \ket{\mathbf u, +}$. In Appendix~\ref{app:UR SU(2)}, we show that, for $\ket{\chi}= \ket{\phi} \otimes \ket{\psi_N}$, 
\be\label{eq:commutator SU(2) operator infinite N main}
    \bra{\chi}\Big[\widetilde{F}_N(P_\mathbf{v}),\widetilde{F}_N(P_\mathbf{w})\Big]\ket{\chi} = \frac{i}{8J_\mathsf{A}^2\,(J_\mathsf{A}+1)^2}\bra{\phi}(\mathbf{v}\times\mathbf{w})\cdot\mathbf{R}_\mathsf{A}\otimes(\mathbf{u}\cdot\mathbf{L}_\mathsf{A})^2\ket{\phi}.
\ee
Therefore, in contrast to the ideal QRF-case, the RRFs associated to $P_{\mathbf v}$ and $P_{\mathbf w}$ do not commute in general in the $N \to \infty$ limit. In fact, the commutator is independent of $N$. Note, however, that the expectation value of Eq.~\eqref{eq:commutator SU(2) operator infinite N main} scales as $1/J_\mathsf{A}$, and therefore vanishes in the limit of very large reference frames.

For a reference frame with $J_\mathsf{A}=1/2$, Eq.~\eqref{eq:commutator SU(2) operator infinite N main} reads 
\be
    \label{eq:trivialsu2bound}\bra{\chi}\Big[\widetilde{F}_N(P_\mathbf{v}),\widetilde{F}_N(P_\mathbf{w})\Big]\ket{\chi} \; = \; \frac{i}{36}\mathrm{Tr}\rho_\phi(\mathbf{v}\times\mathbf{w})\cdot\mathbf{\sigma}_{\mathsf{A},R}, \qquad \text{with} \qquad  \rho_\phi=\tr_L{\ketbra{\phi}{\phi}}.
\ee
As shown in Appendix~\ref{app:UR SU(2)}, the right hand side of Eq.~\eqref{eq:trivialsu2bound} may vanish. (In particular, it vanishes for $\ket{\phi} = \ket e $.)
However one can consider the tighter Robertson-Schr\"odinger uncertainty relation instead. Being non-trivial, this bound means that, for finite QRF resources, $J_\mathsf{A} < \infty$, the uncertainties of the observables $\widetilde F_N(P_{\mathbf v})$ and $\widetilde F_N(P_{\mathbf w})$ cannot be made arbitrarily small, even in the limit of infinitely large number of i.i.d. system copies. Indeed, we show that, for $J_\mathsf{A} = 1/2$ and $\ket \phi = \ket e$,
\be \label{eq:relativefrequencyheisenbergrelation}  
    \Delta_{\chi}^2\widetilde{F}_N(P_\mathbf{v})\Delta_{\chi}^2\widetilde{F}_N(P_\mathbf{w}) \; \geq \; \frac{1}{324} \; \approx \; 0.003.   
\ee

\subsubsection{Uncertainty relations for relational expectation values: the Heisenberg-Weyl group}\label{subsec:uncertaintyofexpectation}
 Suppose now that both the QRF $\mathsf A$ and the system $\mathsf a$ are particles on a one-dimensional line (see Subsection~\ref{subsec:heisenbergweyl}). A natural symmetry group is then the Heisenberg-Weyl group, generated by the 
 position and momentum operators $X$ and $P$, $[X,P] = i \hbar$. 
Using Eq.~\eqref{eq:relquadmain}, one sees that, for states $\ket{\chi}=\ket \phi \otimes \ket{\psi_N}$, $\Delta^2(\widetilde F_N(X)) = \Delta^2(X)$ and similarly for $P$. In fact, in the large-$N$ limit, 
\begin{equation}\label{eq:commutatorFXFP}
 \Big[\widetilde F_N( X),\widetilde F_N(P) \Big]  = i\hbar,
\end{equation}
leading to 
\begin{equation}\label{eq:xp-uncertainty-relation}
    \Delta_\Psi^2(\widetilde F_N(X))\Delta_\Psi^2(\widetilde F_N( P)) \geq \frac{\hbar^2}{4},
\end{equation}
reminiscent of the $X$-$P$ Heisenberg uncertainty relation of ordinary quantum theory. The proof is given in Appendix \ref{app:UR Heisenberg}.

It might seem surprising that the uncertainty relation \eqref{eq:xp-uncertainty-relation} is independent of the ``size'' of the reference frame, since one would expect standard commutativity to be recovered for an infinitely resourceful QRF. The reason is that the group here is the abstract Heisenberg–Weyl group, whose elements are dimensionless complex numbers $\alpha$
that displace the system in the phase space by the same parameter regardless of physical properties of the QRF and the systems, such as their masses. To capture how the relative ``size'' of the QRF and system affects the uncertainty relation, we now turn to the centrally extended Galilei group.

\subsubsection{Uncertainty relations for relational expectation values: the centrally extended Galilei group}
Let $\mathsf A$ be a single particle on a line, of mass $m_{\mathsf A}$, acting as a QRF for the one-dimensional centrally-extended Galilei group and let the system of interest $\mathsf a$ be a single particle on a line of mass $m_\mathsf a$. The relevant generators are those of translations, $ P_{\mathsf A}$ on $\mathcal H_\mathsf A$ and $ P_{\mathsf a^N}$ on $\mathcal H_{\mathsf a^N}$, and the generators of boosts, $ K_{\mathsf A}$ and $ K_{\mathsf a^N}$, satisfying
\begin{equation}\label{eq:galileicommutator}
[ P_\mathsf A,  K_\mathsf A] = i \hbar m_\mathsf A ,
\end{equation}
and similarly on $\mathcal H_\mathsf a$.

Because the commutation relation \eqref{eq:galileicommutator} is formally very similar to the canonical one, we can compute the $N \rightarrow \infty$ limit uncertainty relation for $\widetilde F_N( P)$ and  $\widetilde F_N(K)$ by a construction analogous to that used for the Heisenberg-Weyl group. In Appendix \ref{app:UR Galilei}, we show that, as $N \rightarrow \infty$,
\begin{equation}\label{eq:galileicommrel}
    [\widetilde F_N( P_\mathsf{a}), \widetilde F_N(K_\mathsf{a})] = i \hbar \frac{m^2_\mathsf{a}}{m_\mathsf{A}},
\end{equation}
and therefore 
\begin{equation}\label{eq:uncertaintyrelationgalilei}
    \Delta_\chi^2(\widetilde F_N( P_\mathsf{a}))\Delta_\chi^2(\widetilde F_N( K_\mathsf{a})) \geq \frac{\hbar^2}{4}\frac{m^4_\mathsf a}{m^2_\mathsf A}
\end{equation}
for any state $\ket \chi$ on $\mathcal{H}_\mathsf A \otimes \mathcal H^{\otimes N}_\mathsf a$. The proof is given in Appendix~\ref{app:uncertaintyrelations}. The bound \eqref{eq:uncertaintyrelationgalilei} is proportional to $m^4_\mathsf{a}/m^2_\mathsf{A}$, quantifying the non-commutativity of the REVs in terms of the relative ``sizes'' of the QRF and system. When the QRF's resources, as quantified by its mass, far exceeds the system's, i.e. $m_\mathsf{a}/m_\mathsf{A} \rightarrow 0$,  REVs commute as expected, and the variance bound becomes trivial.

\section{Violation of Bell inequalities}
\label{sec:bellinequality}
\noindent 
In contrast to the ideal case, non-ideal QRFs may lead to a finite uncertainty in the RRFs and REVs even in the limit of infinitely many system copies (see Section~\ref{sec:uncertaintyrelations}). In other words, the observed relative frequencies and expectation values may remain ``unsharp'' even in this limit. One might think this uncertainty originates from uncontrollable, unknown disturbances of the preparation or measurement devices and attribute those disturbances to ``hidden variables'' that, when taken into account, explain the observed uncertainties.  
Within this view, one could still assume that there exist ``hidden true expectation values and probabilities'', the 
observed uncertainty merely reflecting our ignorance 
about the functioning of the preparation or measurement device. The resulting representation of the state would be fully consistent with the quantum de Finetti theorem~\eqref{eq:de Finetti}.
However, 
we show next that this view cannot be maintained under very natural assumptions. Namely, we show that the uncertainties in the distribution of RRFs and REVs are instead irreducible in a quantum-mechanical sense. We do so by proving a Bell theorem for relational expectation values and relational relative frequencies, demonstrating that different RRFs and REVs cannot be jointly well defined under the assumptions of the theorem.

\subsection{The setup}\label{subsec:thegeneralsetup}
\begin{figure}[ht]
\centering
\includegraphics[width=0.8\linewidth]{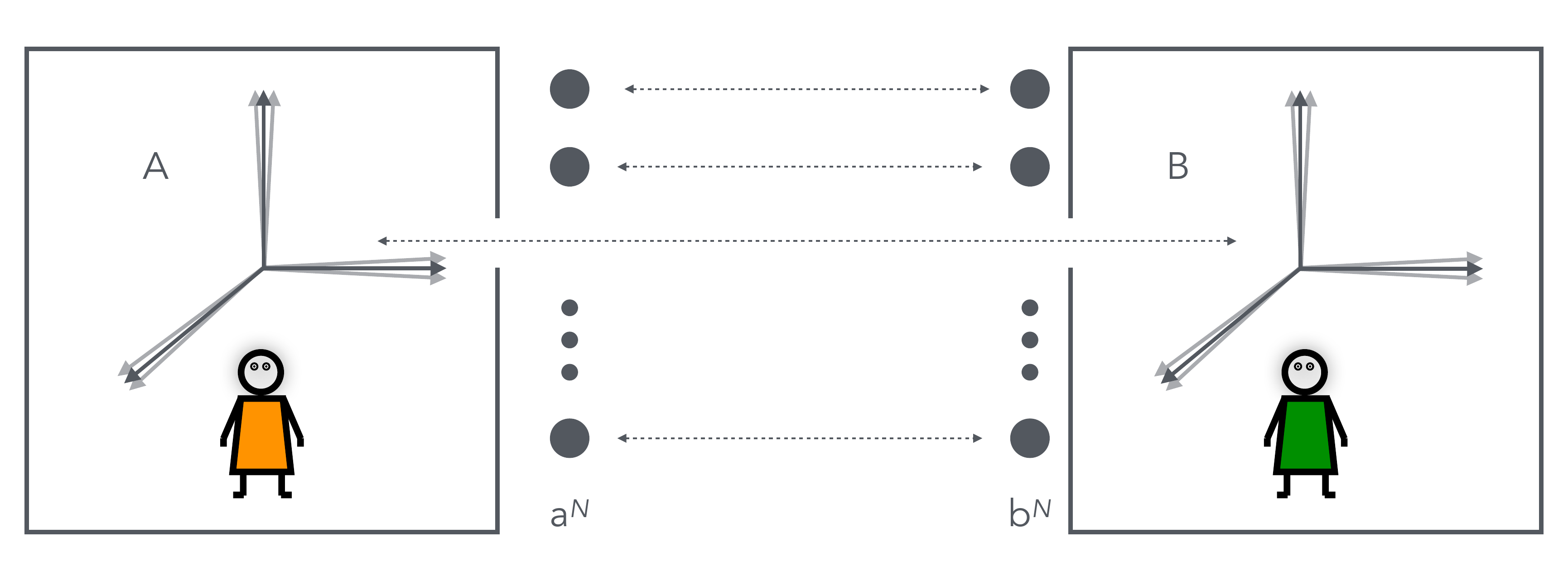}
\caption{\justifying Alice and Bob perform local experiments on $N$ copies of their systems $\mathsf a$ and $\mathsf b$, respectively, measuring relative frequencies and expectation values with respect to their local reference frames $\mathsf A$ and $\mathsf B$, which may be prepared in a quantum state. Because these frames are non-ideal, the relative frequencies and expectation values they measure may have a finite uncertainty even in the limit of arbitrarily large $N$. Here we show that these uncertainties are irreducible in a quantum mechanical sense. That is, the correlations between the RRFs or REVs measured by Alice and Bob may violate Bell inequalities in the large-\(N\) limit.} \label{fig:bellscenario}
\end{figure}
Imagine two observers, Alice and Bob, performing local experiments on $N$ copies of a quantum system, as depicted in 
Fig.~\ref{fig:bellscenario}. In each run, each observer's local measurement consists of choosing a setting and obtaining an outcome, and Alice's and Bob's measurements are assumed to be spacelike separated. 

We assume Alice and Bob 
lack access to any external reference frame for a group $G$. However, each of them has access to a local, possibly non-ideal QRF, denoted \(\mathsf A\) for Alice and \(\mathsf B\) for Bob. These frames are used to determine the relative frequencies and relative expectation values of measurements on the local systems \(\mathsf a\) and \(\mathsf b\). As in the rest of the paper, we consider a single copy of the QRFs $\mathsf A$ and $\mathsf B$ for each $N$ copies of the systems $\mathsf a$ and $\mathsf b$. 
This leads naturally to the PVM construction, where Alice and Bob perform projective measurements corresponding to the REV and the RRF operators (see Section~\ref{rrfs}). 

The total Hilbert space is then
\(
\mathcal H = \mathcal H_{\mathsf A} \otimes \mathcal H_{\mathsf B} \otimes \mathcal H_{\mathsf{a}^N} \otimes \mathcal H_{\mathsf{b}^N},
\)
where \(\mathcal H_{\mathsf{a}^N} \cong \mathcal H_{\mathsf a}^{\otimes N}\), and \(N\) denotes the number of systems \(\mathsf a\) and \(\mathsf b\) (we will later consider the limit \(N \to \infty\)). 
We assume that \(\mathcal H_{\mathsf A} \cong \mathcal H_{\mathsf B} \cong \mathcal H_{\mathrm{QRF}}\) and \(\mathcal H_\mathsf a \cong \mathcal H_\mathsf b \cong \mathcal H_{\mathrm{system}}\).

In their local labs, Alice and Bob can measure either RRFs or REVs. For RRFs, Alice performs PVMs corresponding to observables $\widetilde F_N({P})$, with $P$ a projector on $\mathcal H_\mathsf a$. For REVs, Alice performs PVMs corresponding to observables of the form $\widetilde F_N(A)$, where $A$ is an observable (Hermitian operator) on $\mathcal H_\mathsf a$. Bob proceeds analogously. As both RRFs and REVs have multiple outcomes, to define a Bell experiment we let Alice and Bob measure ``binned'' dichotomic observables built from these multiple-outcome operators. 

We treat the two cases, RRFs and REVs, in a unified way throughout the rest of this section. As in subection~\ref{subsec:heisenbergtradeoff}, we write $ A = \sum_n a_n P_n$, recovering the RRF case when all but one of the $a_n$ vanish and the remaining coefficient is set to unity. Alice then measures a binned observable
\begin{equation}
    \sgn(\widetilde F_N( A) -\mu) = \sum_\lambda \sgn(\lambda -\mu) \Pi_\lambda,
\end{equation}
where $\mu = 1/2$ for RRFs and $0$ for REVs, $\lambda$ are the eigenvalues of $F_N(A)$ and $\Pi_\lambda$ are the corresponding spectral projectors. Here, $\sgn(x) = 1$ for $x > 0$ and $\sgn(x) = -1$ otherwise.

Considering a similar measurement on Bob's side, we define the Bell operator
\begin{equation}
 C\;=\;\sum_{ a, b\,=\,0}^{1}(-1)^{ab}\,E( A_{a}, B_{b}),
\end{equation}
where $E({A}_{ a}, B_{b})\;:=\;
\sgn\!\bigl(\widetilde F_N({A}_{a})-\mu\bigr)\,\otimes\,
\sgn\!\bigl(\widetilde F_N({ B}_{b}) -\mu \bigr)$. Here, $A_{0}$, and $A_1$ specify observables on Alice's side and $ B_0$, and $B_1$ play the same role for Bob. Although 
$C$ is defined in terms of quantum-mechanical operators, its expectation value $\langle C \rangle$ has meaning independently of quantum theory: it is fixed operationally by the experimental results of Alice and Bob. We may therefore ask whether any theory, possibly beyond quantum mechanics, can account for the value of $\langle C \rangle$ in terms of the hidden variables mentioned above -- and in particular, whether a theory satisfying Bell's notion of \emph{local causality} is consistent with Alice's and Bob's observations. 

In a theory satisfying local causality, the joint probabilities admit a decomposition of the form
\begin{equation} \label{eq:localcausalitymain}
p(x,y\mid  a,  b)
=
\int d\lambda\, \rho(\lambda) \, p(x\mid a,\lambda)\, p(y\mid b,\lambda),
\end{equation}
where \( a\) denotes Alice's possible measurement setting, \( b\) is Bob's, and \(x,y = \pm 1\) are the corresponding outcomes of the experiment. The hidden variable is denoted \(\lambda\), and \(\rho(\lambda)\) is its probability distribution. The assumption is that, conditioned on \(\lambda\), Alice's outcome probabilities depend only on Alice's local setting, and Bob's only on Bob's local setting.

 From Eq.~\eqref{eq:localcausalitymain} one can derive Bell-type inequalities~\cite{Bell1964}. Namely, one can show that any local-causal theory satisfies the CHSH inequality~\cite{ClauserHorneShimonyHolt1969}
\begin{equation}
\vert \langle C \rangle \vert \le 2.
\label{eq:CHSH}
\end{equation}
In the following subsections we present two concrete cases -- one for RRFs, one for REVs -- in which there exist operators $A_a$ and $B_{b}$ and a state $\ket{\Psi} \in \mathcal H$ such that 
\begin{equation}
\langle  C \rangle_{\Psi} > 2. 
\end{equation}
This shows that, when measured relative to a non-ideal QRF, RRFs and REVs are incompatible with local causality.

\subsection{Bell's theorem for relational relative frequencies}

 We consider first the case where Alice and Bob measure relative frequencies with respect to their local reference frames. We take the symmetry group to be $\SU(2)$ and the non-ideal QRFs $ \mathsf A$ and $\mathsf B$ are given by the model of Subsection~\ref{subsec:rrfforsu2}. The systems $\mathsf{a}_1,\dots,\mathsf{a}_N$ and $\mathsf{b}_1,\dots,\mathsf{b}_N$ are each a spin-$1/2$ particle. Recall that Alice's measurement of relative frequencies corresponds to a PVM defined by the RRF observable of Eq.~\eqref{eq:niceformofsu2relativefrequencies}.

Alice then binarizes her outcomes according to the dichotomic observable  $\sgn(\widetilde F_N(P_{\mathbf v}) - 1/2)$ for different values of $\mathbf v$. Roughly, this observable counts how many frequencies, in the PVM construction, lie above $1/2$ and how many lie below $1/2$. 

With Bob doing a similar measurement at his end, we define the Bell operator
\begin{equation}
C = \sum_{ a,  b = 0}^1 (-1)^{ ab} E(\mathbf v_{ a}, \mathbf w_{b}),
\end{equation}
where $E(\mathbf v_{a}, \mathbf w_{ b}) = \sgn(\widetilde F_N(P_{\mathbf v_{a}}) - 1/2) \otimes \sgn(\widetilde F_N(P_{\mathbf w_{b}}) - 1/2)$. Here, $\mathbf v_0$ and $\mathbf v_1$ specify directions on Alice's side, while $\mathbf w_0$ and $\mathbf w_1$ do the same on Bob's side. 

Consider the state $\ket{\Psi} = \ket{\Phi}_{\mathsf A \mathsf B}\otimes\ket{\psi, \psi}_{\mathsf a \mathsf b}^{\otimes N}$, where $\ket{\Phi}_\mathsf{\mathsf A \mathsf B}$ is of the form $\ket{\Phi}_\mathsf{AB}=\ket{\phi}_{R_\mathsf{AB}}\ket{\mathbf{\omega},+}_{L_\mathsf{A}}\ket{\mathbf{\omega},+}_{L_\mathsf{B}}$, with the left sectors in an arbitrary product state but the right sectors $R_{\mathsf A \mathsf B} = R_\mathsf A R_\mathsf B$ possibly entangled, and the spin-$1/2$ state $\ket{\psi}$ arbitrary. In Appendix~\ref{app:relationalrelfrequancies} we show that, for all $N$,
\begin{equation}
\bra{\Psi}{C}\ket{\Psi}\;=\;
\sum_{ a, b=0}^{1}(-1)^{ab}\,
\bra{\phi}\sgn(\mathbf v_{ a}\cdot \mathbf R_\mathsf A)\otimes\sgn(\mathbf w_{ b}\cdot \mathbf R_\mathsf B)\ket{\phi}_{R_{\mathsf A \mathsf B}},
\label{eq:CHSH-limit-spins}
\end{equation}
where $\mathbf R_\mathsf A$ is the right-regular angular momentum operator on $\mathcal H_\mathsf A$, and likewise for $\mathsf B$. This reduces the problem to a standard CHSH scenario for two spin-$J$ QRFs, with $J=J_\mathsf{A}=J_\mathsf{B}$. Setting $J = 1/2$, we choose the standard CHSH directions $\mathbf v_0 = \mathbf z = (0,0,1)$ and $\mathbf v_1 = \mathbf x = (1,0,0)$ on Alice's side and $\mathbf w_0 = (1/\sqrt 2)(1,0,1)$ and $\mathbf w_1 = (1/\sqrt 2)(-1,0,1)$ on Bob's. The singlet state
\begin{equation}
\ket{\phi}_{R_{\mathsf A \mathsf B}} = \frac{1}{\sqrt 2}(\ket{ \mathbf z, +}_{R_\mathsf A} \ket{ \mathbf z, -}_{R_\mathsf B} - \ket{\mathbf z, -}_{R_\mathsf A}  \ket{\mathbf z, +}_{R_\mathsf B}), \qquad \mathbf z\cdot \mathbf \sigma \ket{\mathbf z, \pm } = \pm \ket{\mathbf z, \pm},
\end{equation}
leads to $\bra \Psi C \ket \Psi = 2\sqrt2$, thus violating maximally the CHSH inequality for RRFs. 

\subsection{Bell's theorem for relational expectation values}
\label{subsec:bell4revs}
We now turn to the case where Alice and Bob measure expectation values relative to their local QRFs. The symmetry group $G$ is the Heisenberg-Weyl group and the non-ideal QRFs $\mathsf A$ and $\mathsf B$ are realised by particles on a line with position and momentum operators $X$ and $P$. The local systems $\mathsf{a}^N=\mathsf{a}_1,\dots,\mathsf{a}_N$ and $\mathsf{b}^N=\mathsf{b}_1,\dots,\mathsf{b}_N$ are $N$ i.i.d systems of the same kind as the QRFs (see Subsection~\ref{subsec:heisenbergweyl}). Alice and Bob measure binned dichotomic observables $\sgn(\widetilde F(X_\vartheta))$, corresponding to the quadratures $X_\vartheta = \cos \vartheta X + \sin \vartheta P$. 
The Bell operator reads
\begin{equation}
 C\;=\;\sum_{ a, b\,=\,0}^{1}(-1)^{ab}\,E(\vartheta_{ a},\varphi_{ b}),
\end{equation}
where $E(\vartheta_{ a},\varphi_{ b})=
\sgn\bigl(\widetilde F_N( X_{\vartheta_{ a}})\bigr)\otimes
\sgn\bigl(\widetilde F_N( X_{\varphi_{b}})\bigr)$. Here, $\vartheta_0$, and $\vartheta_1$ specify quadratures on Alice's side and $\varphi_0$, and $\varphi_1$ play the same role on Bob's side.

As shown in Appendix~\ref{app:proofofbell}, there exists a joint state of the form $\ket{\Psi} = \ket{\Phi}_{\mathsf A \mathsf B} \otimes \ket{\gamma, \gamma}_{\mathsf a \mathsf b}^{\otimes N}$ such that, as $N \to \infty$
\begin{equation}
\bra{\Psi}{C}\ket{\Psi}\;=\;
\sum_{a, b=0}^{1}(-1)^{ab}\,
\bra{\Phi}\sgn( X_{\vartheta_{ a}})\otimes\sgn( X_{\varphi_{ b}})\ket{\Phi}_{\mathsf A \mathsf B},
\label{eq:CHSH-limit}
\end{equation}
thus reducing the problem to a standard continuous--variable CHSH inequality for quadrature binned observables. It is well known that there exist angles $\vartheta_0$, $\vartheta_1$, $\varphi_0$ and $\varphi_1$ and states $\ket{\Phi}_{AB}$ such that this inequality is violated. 

In particular, following Refs.~\cite{garcia2004proposal,NhaCarmichael2004}, we set $\vartheta_0 =0$, $\vartheta_1 =\pi/2$, $\varphi_0 = -\pi/4$ and $\varphi_1 = \pi/4$ and  
\begin{equation}
\ket{\Phi}_{\mathsf A \mathsf B} \propto \sum_{k =0}^\infty (k + 1) r^k \ket{k}_\mathsf A\otimes \ket{k}_\mathsf B
\end{equation}
where $\ket{k}$ is the Fock basis and $r \approx 0.572$. This state is not optimal for the violation, but it is a feasible one, obtained by subtracting one excitation from each subsystem of a two-mode squeezed state. Applying this choice of settings and state $\ket{\Phi}_{\mathsf A \mathsf B}$ to our problem, we obtain, in the large $N$ limit, a violation of the inequality by $\bra{\Psi} C \ket{\Psi} \approx 2.048$.

\section{Experimental proposal of a Bell experiment for relational expectation values}
\noindent 

\noindent We now present a quantum-optical experimental proposal to implement the Bell experiment for REVs (Subsection~\ref{subsec:bell4revs}), based on local measurements of the continuous-variable relational expectation value observables of Eq.~\eqref{eq:relquadmain}. The
experimental ingredients are standard in pulsed continuous-variable optics: balanced homodyne readout of optical quadratures~\cite{LvovskyRaymer2009,HansenEtAl2001,AppelEtAl2007},
coherent linear-optical transformations on many modes~\cite{ReckZeilingerBernsteinBertani1994,ClementsEtAl2016},
and temporal-mode control of pulsed light~\cite{BrechtReddySilberhornRaymer2015,AnsariEtAl2018,SerinoEtAl2023}. Figure~\ref{fig:pvm-experimental-sequence-main} shows how these measurements are implemented in our quantum-optical proposal; see Appendix~\ref{app:pvm-quantum-optics} for a detailed explanation.

\begin{figure}[t]
    \centering
    \includegraphics[width=0.76\linewidth]{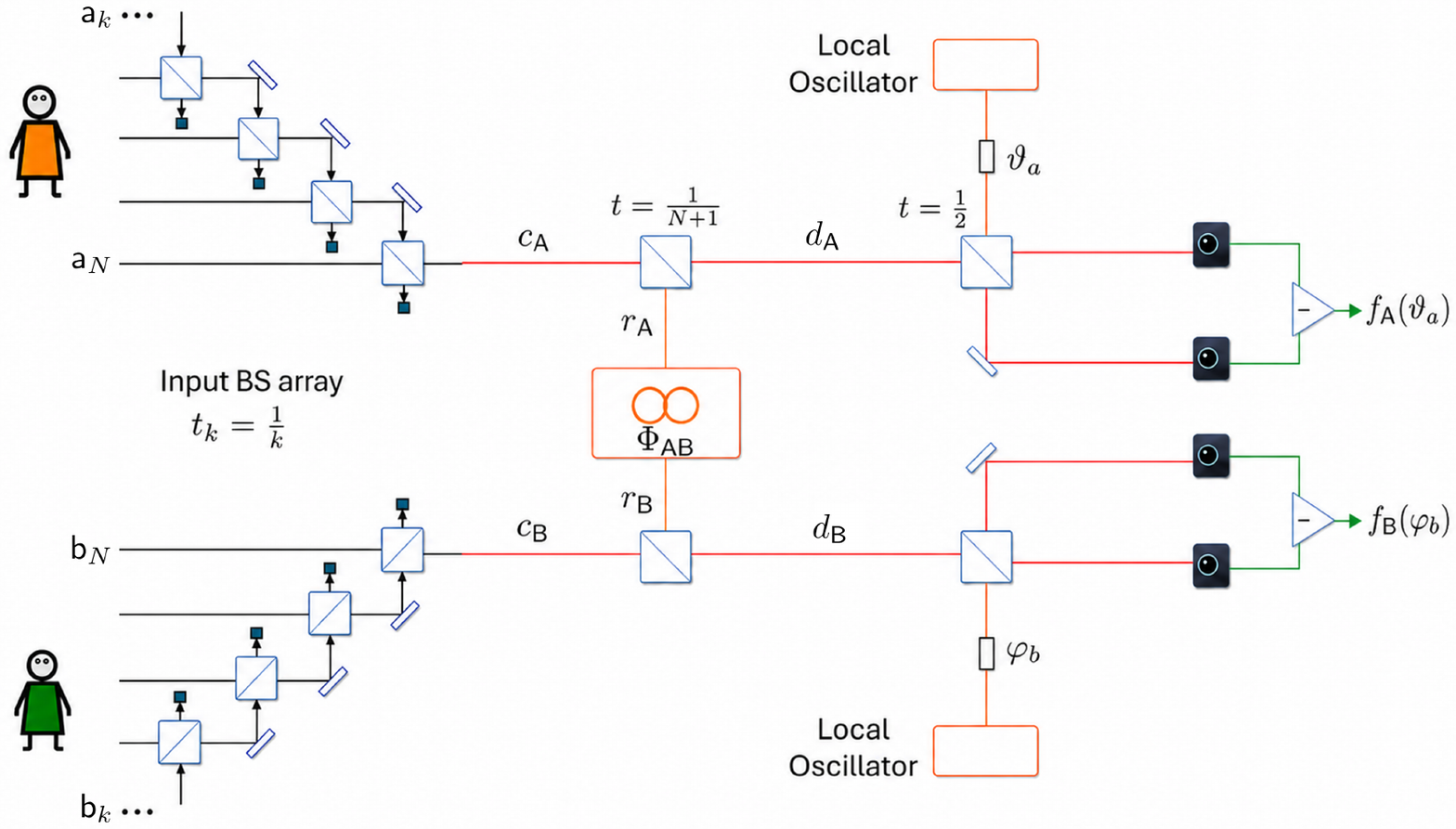}
    \caption{
Quantum-optical implementation of the Bell experiment for relational expectation values. Alice and Bob receive a block of \(N\) signal modes, \(\mathsf{a}_1,\ldots,\mathsf{a}_N\) and \(\mathsf{b}_1,\ldots,\mathsf{b}_N\), respectively. The input beam-splitter arrays coherently combine each block into the collective modes \(c_\mathsf A=N^{-1/2}\sum_{i=1}^N \mathsf{a}_i\) and \(c_\mathsf B=N^{-1/2}\sum_{i=1}^N \mathsf{b}_i\). With the convention used in the figure, the \(k\)-th beam splitter in the input array has intensity transmissivity \(t_k=1/k\), for \(k=2,\ldots,N\), for the newly added input mode into the continuing collective output mode. The collective modes \(c_\mathsf A,c_\mathsf B\) are mixed locally with the quantum-reference-frame modes \(r_\mathsf A,r_\mathsf B\), prepared in the entangled state \(\Phi_{\mathsf A \mathsf B}\). The unbalanced beam splitters are chosen, as in Eq.~(E3), such that the output modes measured by homodyne detection are \(d_\mathsf A=(N+1)^{-1/2}c_\mathsf A-\sqrt{N/(N+1)}\,r_\mathsf A\) and \(d_\mathsf B=(N+1)^{-1/2}c_\mathsf B-\sqrt{N/(N+1)}\,r_\mathsf B\). Thus the transmissivity of the collective signal mode into the measured output port is \(t=1/(N+1)\). The homodyne phases \(\vartheta_a\) and \(\varphi_b\) select the local quadrature settings. Up to the positive rescaling factor \(s_N=\sqrt{1+1/N}\), the measured homodyne outcomes implement the relational quadratures \(N^{-1}\sum_{i=1}^N X_{\mathsf{a}_i}(\vartheta_a)-X_{r_\mathsf{A}}(\vartheta_a)\) and \(N^{-1}\sum_{i=1}^N X_{\mathsf{b}_i}(\varphi_b)-X_{r_\mathsf{B}}(\varphi_b)\). The real-valued outcomes are denoted by \(f_\mathsf{A}(\vartheta_a)\) and \(f_\mathsf{B}(\varphi_b)\), and are sign-binned to form the Bell variables \(A_a=\operatorname{sgn} f_\mathsf{A}(\vartheta_a)\) and \(B_b=\operatorname{sgn} f_\mathsf{B}(\varphi_b)\).
}
    \label{fig:pvm-experimental-sequence-main}
\end{figure}

The non-ideal QRF $\mathsf A$ ($\mathsf B$) is realised by the optical mode $r_\mathsf A$ ($r_\mathsf B$), with quadratures $X_{\vartheta, \mathsf A} = (\mathrm{exp}(-i \vartheta) r_\mathsf A + \mathrm{exp}(i \vartheta) r^\dagger_{\mathsf A})/\sqrt 2$ ($X_{\varphi, \mathsf B} = (\mathrm{exp}(-i \varphi) r_\mathsf B + \mathrm{exp}(i \varphi) r^\dagger_{\mathsf B})/\sqrt 2$), and the two $r_\mathsf A$ and $r_\mathsf B$ may be prepared in a fixed entangled state $\Phi_{\mathsf{AB}}$. The systems under study are realised by blocks of $N$ signal pulses, $\mathsf{a}_1,\dots,\mathsf{a}_N$ for Alice and $\mathsf{b}_1,\dots,\mathsf{b}_N$ for Bob, with quadratures $X_{\vartheta, \mathsf{ a}_i} =  (\mathrm{exp}(-i \vartheta) \mathsf{a}_i + \mathrm{exp}(i \vartheta) \mathsf{a}^\dagger_{i})/\sqrt 2$ and $X_{\varphi, \mathsf{b}_i} =  (\mathrm{exp}(-i \varphi) \mathsf{b}_i + \mathrm{exp}(i \varphi) \mathsf{b}^\dagger_{i})/\sqrt 2$. The pulses are prepared in an i.i.d. state $\ket{\alpha}^{\otimes N}$. As in the rest of the paper, both Alice and Bob use only one copy of their QRF per $N$ copies of their local systems. This means that the same preparation of the QRF is used for all relative quadrature measurements; the choices of $X_\vartheta$ and $X_\varphi$ are made only by the homodyne readout phase. 
The goal for Alice (Bob) is to do a projective measurement corresponding to the REV operator $\widetilde F_{N, \mathsf A}(X_\vartheta)$ ($\widetilde F_{N, \mathsf B}(X_\varphi)$). These operators can be expressed as ``relational quadratures''  
\begin{equation}\label{eq:relationalquadratures}
\widetilde F_{N, \mathsf A}(X_\vartheta) \; = \; \frac{1}{N}\sum_{i =1}^N X_{\mathsf a_i}(\vartheta) - X_\mathsf A(\vartheta), \qquad \widetilde F_{N, \mathsf B}(X_\varphi) \; = \; \frac{1}{N}\sum_{i =1}^N X_{\mathsf b_i}(\varphi) - X_\mathsf B(\varphi).
\end{equation}

The implementation is the following. After receiving their signal modes, Alice and Bob respectively prepare the collective signal mode
\begin{equation}
     c_\mathsf{A}
    :=
    \frac1{\sqrt N}\sum_{i=1}^N \mathsf{a}_i,
    \qquad
    c_\mathsf{B}
    :=
    \frac1{\sqrt N}\sum_{i=1}^N  \mathsf{b}_i.
    \label{eq:bipartite-collective-modes-main}
\end{equation}
Each collective mode $c_\mathsf{A}$, $c_\mathsf{B}$ is then mixed with the corresponding QRF
mode $r_\mathsf{A}$, $r_\mathsf{B}$, respectively, on an unbalanced beamsplitter (see Fig.~\ref{fig:pvm-experimental-sequence-main}). Then,  as shown in Appendix~\ref{app:pvm-quantum-optics}, a homodyne phase \(\vartheta\) on
Alice's side and \(\varphi\) on Bob's side implements the local PVMs of the relational quadratures of Eq.~\eqref{eq:relationalquadratures} up to a common positive rescaling factor $s_N$. 
Since this factor is positive, it does not affect sign binning.

Then, for each experimental run, Alice and Bob choose local homodyne phases
\(\vartheta_a\) and \(\varphi_b\), record the real-valued outcomes corresponding to their relational PVM,
\(f_\mathsf{A}(\vartheta_a)\) and \(f_\mathsf{B}(\varphi_b)\), and binarize the outcomes as
\begin{equation}
    A_a:=\sgn f_\mathsf{A}(\theta_a),
    \qquad
    B_b:=\sgn f_\mathsf{B}(\phi_b).
    \label{eq:bipartite-binarization}
\end{equation}
The correlators \(E_{ab}:=\langle A_aB_b\rangle\) are then combined into the
CHSH expression. Thus, for a fixed $N$ the optical module implements, on each side, the local relational PVM
assumed in the Bell construction of Subsection~\ref{subsec:bell4revs}. As $N$ grows, the proposed experiment approaches the theoretical analysis.

\section{Discussion}
\noindent Historically, successive physical theories have often extended their predecessors while recovering them as limiting or special cases. At the same time, the formal structure of prediction has grown progressively weaker. In classical mechanics, specifying a state allows, in principle, a deterministic prediction of the outcomes of consecutive measurements. In classical statistical mechanics, predictions rest on classical probabilities that reflect incomplete knowledge of the underlying state \cite{Laplace1814Probabilities,Gibbs1902StatisticalMechanics}. In quantum mechanics, the predictive content is given by quantum probabilities assigned to measurement outcomes through the Born rule \cite{Born1926Quantenmechanik}. Bell and Kochen--Specker type results show that these probabilities cannot, in general, be interpreted as ignorance about pre-existing values assigned noncontextually to all observables \cite{Bell1964,KochenSpecker1967}.

This sequence suggests a natural question. If classical statistical mechanics replaces deterministic predictions by classical probability distributions, and quantum mechanics replaces those distributions by quantum probability assignments, might a yet deeper theory further weaken the status of probabilities themselves? The question is closely related to the sense in which quantum theory might itself be a special case of a broader framework. In recent years, the problem has been addressed through reconstructions of quantum theory from physically motivated axioms~\cite{Hardy2001QuantumTheoryAxioms,DakicBrukner2011QuantumTheory,MasanesMuller2011Derivation,ChiribellaDArianoPerinotti2011Informational,HohnWever2017QuantumTheoryQuestions}, often formulated within the framework of generalized probabilistic theories~\cite{Barrett2007GPT}. Yet while several such reconstructions exist, no comparably complete nonclassical and nonquantum theory with a rich dynamical structure is known. The gap suggests that, as long as a system's state is identified with a complete set of sharp probabilities for all admissible measurements, a physically compelling extension of quantum theory may be hard to obtain.

The present paper realizes this possibility in a concrete operational setting within the quantum framework.
We start from the standard connection between Born probabilities and
relative frequencies. For ideal repetitions of an experiment, the relative
frequency operator \(F_N(P)\) associated with a projector \(P\) becomes sharp on product states in the large-\(N\) limit, with limiting value given by the Born probability. Ordinary relative frequencies for different measurements are therefore jointly well defined in that limit, even when the underlying single-system observables are incompatible. We then relax the
implicit assumption that the reference frame defining the measurement is ideal. Once the reference frame is itself treated as a finite quantum system,
the operationally accessible quantity is not \(F_N(P)\), but its relational,
\(G\)-twirled counterpart \(\widetilde F_N(P)\). We showed that such
relational relative frequency operators need not commute, even as \(N\to\infty\). Thus the limiting relative frequencies themselves can become incompatible observables. Although we focus in the limit $N \to \infty$, it is important to note that, in any realistic case where $N$ is large but  finite, the total number of degrees of freedom, including system copies and reference frames, is bounded. Operationally, one may then choose how to divide these degrees of freedom between reference frames and systems; however, the consequences of our analysis will persist. 

This non-commutativity has several consequences. First, it yields uncertainty relations for relational relative frequencies and expectation values: finite reference frames imply that the relative frequencies associated with incompatible measurements cannot both be made arbitrarily sharp. In the explicit examples above, the resulting lower bounds are controlled by the finite resources of the frame and vanish in the corresponding ideal limit. Second, we distinguish two operational constructions. In the PVM construction, one measures the spectral projectors of the Hermitian RRF operator itself; in the POVM construction, one instead relativises the full finite-\(N\) frequency POVM. Their first moments agree, but higher moments need not: the (m)-th moment of the measured POVM outcomes is not generally the expectation value of \(\widetilde F_N(E)^m\). Moreover, in the PVM construction, we derive the corresponding outcome statistics and show that the measured variance separates into an ordinary \(1/N\) sampling term and a residual reference-frame contribution that can survive in the large-\(N\) limit. These effects could be observed in a homodyne proposal -- a concrete quantum-optical route to detecting the predicted residual block-to-block fluctuations.

Most importantly, by relaxing the notion of probability itself and allowing for ``probabilities of probabilities'', we obtain a framework that can still be expressed in the mathematical language of quantum theory -- Hilbert spaces, states, operators, and POVMs -- yet goes beyond the textbook idealization in which probabilities are sharp quantities defined relative to ideal, infinitely resourceful reference frames. Here, the uncertainty in assigning sharp probabilities to a system is not merely ignorance about underlying ``true'' quantum probabilities. The Bell-type argument shows that the limiting probabilities for relative frequencies cannot in general be completed by pre-existing sharp values. Their indefiniteness is therefore fundamental and irreducible -- in the same operational sense in which ordinary quantum probabilities cannot be reduced to ignorance about pre-existing measurement outcomes.

A natural place where the assumption of finite reference-frame resources
becomes unavoidable is quantum gravity~(see, e.g. \cite{rovelli1991observable,rovelli1991quantum,Giddings_2006, girelli2008quantum}). A recurring lesson from black-hole
thermodynamics and holography is that a bounded spacetime region is expected
to contain only a finite amount of physically accessible information
\cite{Bekenstein1973BlackHolesEntropy,Hawking1975ParticleCreation,Hooft1993DimensionalReduction,Susskind1995WorldHologram,Bousso2002HolographicPrinciple}. 
The strongest indication is the Bekenstein--Hawking area law: the maximal
entropy associated with a region enclosed by a surface of area \(A\) is
\(S_{\max}=A/(4l_P^2)\), where \(l_P=\sqrt{G\hbar/c^3}\) is the Planck
length. Equivalently, the maximal information content measured in bits is
\(I_{\max}=A/(4l_P^2\ln 2)\). If this bound is translated into an effective
Hilbert-space dimension, one obtains
\(\dim\mathcal H_{\max}=2^{I_{\max}}=\exp(A/(4l_P^2))\). Thus a finite
boundary area suggests a finite number of physically distinguishable degrees
of freedom.

Operationally, the bound means that a finite laboratory region cannot contain arbitrarily many perfectly distinguishable memories, clocks, reference frames, or measurement records. Attempts to store ever increasing amounts of information in a fixed region are expected to produce significant gravitational back-action. The holographic entropy bound thus supports the view that non-ideal reference frames are not merely a practical consequence of using finite resources inside an otherwise infinitely resourceful spacetime region. Rather, their non-ideality is fundamental and irreducible: a bounded region has only finite information capacity and therefore cannot support an ideal, infinitely sharp frame. In the presence of finite resources, then, the concept of quantum state -- and, more generally, of a quantum measurement and quantum operation -- may have to be extended to incorporate the full probability distribution of relational relative frequencies; and in scenarios involving finite spacetime regions in quantum gravity, an even more radical revision of the notion of state may be required.

\begin{acknowledgments}
{\v C}.B. thanks Borivoje Daki\'c for years-long discussions on how the notion of predictability can be extended beyond probabilities. {\v C}.B. acknowledges financial support from the Austrian Science Fund
(FWF) [10.55776/COE1] (Quantum Science
Austria) and [10.55776/F71] (BeyondC), as well as the ID\#~62312 `Quantum Information Structure of Spacetime (QISS)' and ID\#~63683  `Without Spacetime (WOST)' grants from the John Templeton Foundation. 

\end{acknowledgments}

\bibliography{PR}

\appendix
\section{Uncertainty relations for RRFs}
\label{app:uncertaintyrelations}
\subsection{Heisenberg-Weyl} \label{app:UR Heisenberg}
Here we derive the large-$N$ commutation relations~\eqref{eq:commutatorFXFP} and~\eqref{eq:galileicommrel}. 

We start with Eq.~\eqref{eq:commutatorFXFP} and the case of two particles on a line. The Hilbert space is $\mathcal{H} = \mathcal H_{\mathsf A} \otimes \mathcal H^{\otimes N}_{\mathsf a}$, where $\mathsf{A}$ labels the QRF particle and $\mathsf{a}$ labels a single system particle. For simplicity we identify the Hilbert spaces by their ordering, omitting the labels.

By definition 
\begin{align}
\widetilde F_N(X) =& \int \frac{\mathrm{d}^2 \alpha}{\pi} \, \ketbra{\alpha}{\alpha} \otimes D_\alpha^{\otimes N} F_N (X)  D_\alpha^{\dagger \otimes N} \label{eq:tildeF_NX},\\
\widetilde F_N(P) =& \int \frac{\mathrm{d}^2 \alpha}{\pi} \, \ketbra{\alpha}{\alpha} \otimes D_\alpha^{\otimes N} F_N (P)  D_\alpha^{\dagger \otimes N} \label{eq:tildeF_NP},
\end{align}
where $X$ and $P$ are the dimensionless position and momentum operators, respectively. Using the linearity of $F_N(X)$ and $F_N(P)$ together with $D_\alpha X D^\dagger_\alpha = X - \sqrt 2 \mathrm{Re}(\alpha) \id \  \text{and} \ 
D_\alpha P D^\dagger_\alpha = P - \sqrt 2 \mathrm{Im}(\alpha) \id$, and the resolution of identity for coherent states, Eqs.~\eqref{eq:tildeF_NX} and~\eqref{eq:tildeF_NP} simplify to 
\begin{align}
\widetilde F_N(X) =& - \int \frac{\mathrm{d}^2 \alpha}{\pi} \, \sqrt{2} \mathrm{Re}(\alpha) \, \ketbra{\alpha}{\alpha} \otimes \id + \id \otimes F_N(X) \\
\widetilde F_N(P) =& - \int \frac{\mathrm{d}^2 \alpha}{\pi} \, \sqrt{2} \mathrm{Im}(\alpha) \, \ketbra{\alpha}{\alpha} \otimes \id + \id \otimes F_N(P).
\end{align}
The only term in the commutator~\eqref{eq:commutatorFXFP} containing either $F_N(X)$ or $F_N(P)$ is $[\id \otimes F_N(X), \id \otimes F_N(P)]$. This term scales as $1/N$, so it vanishes in the limit $N \rightarrow \infty$. We thus concentrate only in the term that is non-trivial in the reference frame.

Now, one can see that
\begin{align}
\int \frac{\mathrm{d}^2 \alpha}{\pi} \, \sqrt{2} \mathrm{Re}(\alpha) \, \ketbra{\alpha}{\alpha}=& \; X, \label{eq:coherentstateidentityX}\\
\int \frac{\mathrm{d}^2 \alpha}{\pi} \, \sqrt{2} \mathrm{Im}(\alpha) \, \ketbra{\alpha}{\alpha}=& \; P, \label{eq:coherentstateidentityP} 
\end{align}
as consequences of well-known identities of coherent states (see e.g. Ref.~\cite{cahill1969ordered}). For completeness, in the following we give an elementary proof of~\eqref{eq:coherentstateidentityX}. A completely analogous proof holds for~\eqref{eq:coherentstateidentityP}.

Recall the resolution of the identity for coherent states,
\begin{equation}
    \frac{1}{\pi}\int \mathrm d^2 \alpha  \, \lvert\alpha\rangle\langle\alpha\rvert \, = \id,
    \label{eq:resolution}
\end{equation}
where $\mathrm d^2 \alpha = \mathrm d\,\mathrm{Re}(\alpha)\, \mathrm d\,\mathrm{Im}(\alpha)$.
Writing $\mathrm{Re}(\alpha) = (\alpha + \alpha^*)/2$, we have
\begin{equation}
    \frac{\sqrt{2}}{\pi}\int \mathrm d^2 \alpha \;\mathrm{Re}(\alpha)\,
    \lvert\alpha\rangle\langle\alpha\rvert
    = \frac{\sqrt{2}}{2\pi}\int \mathrm d^2 \alpha\,\alpha\,
    \lvert\alpha\rangle\langle\alpha\rvert
    + \frac{\sqrt{2}}{2\pi}\int \mathrm d^2 \alpha\,\alpha^*\,
    \lvert\alpha\rangle\langle\alpha\rvert.
    \label{eq:split}
\end{equation}
We evaluate each integral by inserting the Fock-state expansion
$\lvert\alpha\rangle = e^{-\lvert\alpha\rvert^2/2}\sum_{n=0}^{\infty}
\frac{\alpha^n}{\sqrt{n!}}\lvert n\rangle$ and using the standard Gaussian
integral
$\frac{1}{\pi}\int \alpha^{*m} \alpha^n\,
e^{-\lvert\alpha\rvert^2}\, \mathrm d^2 \alpha = n!\,\delta_{mn}$.
A direct computation then gives
\begin{equation}
    \frac{1}{\pi}\int \mathrm d^2 \alpha\,\alpha\,
    \lvert\alpha\rangle\langle\alpha\rvert
    = \sum_{n=0}^{\infty}\sqrt{n+1}\,\lvert n\rangle\langle n+1\rvert
    = a,
    \qquad
    \frac{1}{\pi}\int \mathrm d^2 \alpha\,\alpha^*\,
    \lvert\alpha\rangle\langle\alpha\rvert
    = a^\dagger.
\end{equation}
Substituting these expressions into Eq.~\eqref{eq:split} and recalling the definition
$X = (a+a^{\dagger})/\sqrt{2}$, we obtain
\begin{equation}
    \frac{\sqrt{2}}{\pi}\int \mathrm d^2 \alpha\;\mathrm{Re}(\alpha)\,
    \lvert\alpha\rangle\langle\alpha\rvert
    = \frac{\sqrt{2}}{2}\bigl(a + a^{\dagger}\bigr)
    = X. \qedhere
\end{equation}
Therefore, 
\begin{align}
\widetilde F_N(X) \; =& \;  \id \otimes F_N(X) - X \otimes \id \label{eq:frequencyX}\\
\widetilde F_N(P) \; =& \; \id \otimes F_N(P) - P \otimes \id. \label{eq:frequencyP}
\end{align}
and
\begin{equation}
\underset{N \rightarrow \infty}{\mathrm{lim}}  [\widetilde F_N(X),\widetilde F_N(P)] = [X \otimes \id, P \otimes \id] = i .
\end{equation}
Going back to units with $\hbar$ we obtain Eq.~\eqref{eq:commutatorFXFP}.
\subsection{Galilei group in 1D} \label{app:UR Galilei}

We now proceed with the centrally-extended Galilei group in one dimension and the commutator~\eqref{eq:galileicommrel}. Given the similarities between the Heisenberg-Weyl and the centrally-extended Galilei group, we can follow here an analogous construction to the one above. In the following we construct ``coherent states'' for the mass-$m$ irreducible representation of the centrally-extended Galilei group.

The one-dimensional algebra is generated by the momentum $P$, the boost generator $K = P t - m X$, and the Hamiltonian $ H$. As we do not consider dynamics, we concentrate only on $P$ and $K$. Their commutation relation reads 
\begin{equation}
[P,K] = i m 
\end{equation}
Here $m$ is the mass, appearing as a central charge. 
We set $\hbar = 1$ throughout this section.

Exploiting the formal analogy between $[P,K] = im$ and $[X,P] = i$ in the Heisenberg--Weyl algebra, we define
\begin{equation}
a = \tfrac{1}{\sqrt{2}}\bigl(P + iK\bigr), \qquad a^\dagger = \tfrac{1}{\sqrt{2}}\bigl(P - iK\bigr),
\end{equation}
which satisfy
\begin{equation}
[a, a^\dagger] = m.
\end{equation}
Note that the mass $m$ is \emph{not} absorbed into the normalisation; this choice is essential for the physical consequences below. The inverse relations are
\begin{equation}
P = \tfrac{1}{\sqrt{2}}(a + a^\dagger), \qquad K = \tfrac{i}{\sqrt{2}}(a^\dagger - a).
\end{equation}
We define the ``vacuum'' seed state $\ket{0}$ by $a\ket{0} = 0$. In particular, at $t=0$, when the boost generator is $K(0) =  -mX$, the position representation of this state reads
\begin{equation}
\psi_0(x) \propto \exp\!\left(-\tfrac{m}{2}\,x^2\right) \qquad 
\text{ with variance } \qquad \Delta x = \frac{1}{\sqrt{2m}}.
\end{equation}
Note that the width $\Delta x$ shrinks as $m \to \infty$.

We can then construct Galilei coherent states as group displacements of the seed state,
\begin{equation}
\ket{\alpha} = D(\alpha)\ket{0}, \qquad D(\alpha) = \exp\!\bigl(\alpha\, a^\dagger - \alpha^*\, a\bigr).
\end{equation}
Equivalently, parametrising $\alpha = -(\eta + i\xi)/\sqrt{2}$ in terms of a translation $\xi$ and a boost velocity $\eta$, we have
\begin{equation}
D(\alpha) = \exp\!\bigl[-i\,(\xi P + \eta K)\bigr],
\end{equation}
so that coherent states are obtained by acting on the vacuum with a combined translation and boost.

The coherent states are eigenstates of $a$ with eigenvalue $m\alpha$,
\begin{equation}
a\ket{\alpha} = m\alpha \ket{\alpha},
\end{equation}
and obey the (mass-dependent) overlap
\begin{equation}
\langle \beta \vert  \alpha \rangle = \exp\!\left[\, m\!\left(\beta^*\alpha - \tfrac{1}{2}\vert \alpha\vert ^2 - \tfrac{1}{2}\vert \beta\vert ^2\right)\right],
\end{equation}
whose modulus squared
\begin{equation}
\bigl\vert \langle \beta \vert  \alpha \rangle\bigr\vert ^2 = \exp\!\bigl[-m\,\vert \alpha - \beta\vert ^2\bigr]
\end{equation}
vanishes for any $\alpha \neq \beta$ in the limit $m \to \infty$. The resolution of the identity reads
\begin{equation}
\id = \frac{m}{\pi} \int \mathrm d^2 \alpha\; \ket{\alpha}\!\bra{\alpha},
\end{equation}
which exhibits a normalisation factor $m/\pi$ playing the role of the usual $1/\pi$ in the Heisenberg-Weyl case.

From $D^\dagger a D = a + m\alpha$ and $D^\dagger a^\dagger D = a^\dagger + m\alpha^*$, the displacement acts on the generators as
\begin{align}
D(\alpha)\, P \, D(\alpha)^\dagger &= P - m\sqrt{2}\,\mathrm{Re}\,\alpha, \\
D(\alpha)\, K \, D(\alpha)^\dagger &= K - m\sqrt{2}\,\bigl(t\,\mathrm{Re}\,\alpha + \mathrm{Im}\,\alpha\bigr).
\end{align}

Consider now two particles, the QRF $\mathsf A$ and a single system $\mathsf a$. The Hilbert space is $\mathcal{H} = \mathcal H_{\mathsf A} \otimes \mathcal H_{\mathsf a}$, where $\mathsf A$ labels the QRF and $\mathsf a$ labels a single system. Both $\mathsf A$ and $\mathsf a$ carry irreducible representations of the centrally-extended Galilei group with masses $m_{\mathsf A}$ and $m_{\mathsf a}$, respectively, and have their own coherent states $\ket{\alpha}_{\mathsf A}$ and $\ket{\alpha}_{\mathsf a}$ (with the corresponding mass in the commutator $[a,a^\dagger] = m$). 

As for a single particle on a line, the following integrals on the QRF $\mathsf A$ will be needed:
\begin{equation}\label{eq:rmoments}
\frac{m_{\mathsf A}}{\pi}\!\int \mathrm d^2 \alpha\;\sqrt{2}\,\mathrm{Re}\,\alpha\; \ket{\alpha}\!\bra{\alpha}_{\mathsf A} = \frac{P_{\mathsf A}}{m_{\mathsf A}}, \qquad
\frac{m_{\mathsf A}}{\pi}\!\int \mathrm d^2 \alpha\;\sqrt{2}\,\mathrm{Im}\,\alpha\; \ket{\alpha}\!\bra{\alpha}_{\mathsf A} = \frac{K_{\mathsf A}(0)}{m_{\mathsf A}},
\end{equation}
where $ K_{\mathsf A}(0) = -m X_{\mathsf A}$ denotes the boost generator of QRF $\mathsf A$ at $t =0$.

Define the integrals
\begin{equation}
I_{P} = \frac{m_{\mathsf A}}{\pi} \int \mathrm d^2 \alpha \; \ket{\alpha}\!\bra{\alpha}_{\mathsf A} \otimes D(\alpha)\, P_{\mathsf a}\, D(\alpha)^\dagger,
\end{equation}
\begin{equation}
I_{K} = \frac{m_{\mathsf A}}{\pi} \int \mathrm d^2 \alpha \; \ket{\alpha}\!\bra{\alpha}_{\mathsf A} \otimes D(\alpha)\, K_{\mathsf a}\, D(\alpha)^\dagger.
\end{equation}
Using the displacement identities above and Eq.~\eqref{eq:rmoments}, we obtain
\begin{equation}\label{eq:integralP}
\; I_{P} = P_{\mathsf a} - \frac{m_{\mathsf a}}{m_{\mathsf A}}\, P_{\mathsf A} \;
\end{equation}
and, expanding $K_{\mathsf a} = P_{\mathsf a} t - m_{\mathsf a} X_{\mathsf a}$ and $K_{\mathsf A} = P_{\mathsf A} t - m_{\mathsf A} X_{\mathsf A}$, we get
\begin{equation}\label{eq:integralK}
\; I_{K} = K_{\mathsf a} - \frac{m_{\mathsf a}}{m_{\mathsf A}}\, K_{\mathsf A}, \;
\end{equation}
an operator identity valid for all $t$.

We are now ready to compute the final result. Consider $N$ copies of the system $\mathsf a$ and the Hilbert space $\mathcal H = \mathcal H_{\mathsf A} \otimes \mathcal H_{\mathsf a}^{\otimes N}$. The RRF operators on $\mathcal H$ we are interested in are
\begin{align}
\widetilde F_N(P) =& \frac{m_{\mathsf A}}{\pi} \int \mathrm{d}^2 \alpha \, \ketbra{\alpha}{\alpha} \otimes D_\alpha^{\otimes N} F_N (P)  D_\alpha^{\dagger \otimes N} \label{eq:tildeF_NX1}, \\
\widetilde F_N(K) =& \frac{m_{\mathsf A}}{\pi} \int \mathrm{d}^2 \alpha \, \ketbra{\alpha}{\alpha} \otimes D_\alpha^{\otimes N} F_N (K)  D_\alpha^{\dagger \otimes N} \label{eq:tildeF_NP1}.
\end{align}
Expanding $F_N (P)$ and $F_N (K)$ in the equations above and using the integrals $I_{P}$ and $I_{K}$ of Eqs.~\eqref{eq:integralP} and~\eqref{eq:integralK}, we can reduce the commutator~\eqref{eq:galileicommutator} to commutators of the generators on the systems $\mathsf a$ and on the QRF $\mathsf A$. As for the Heisenberg-Weyl group, the terms with non-trivial commutators on the systems $\mathsf a$ vanish in the limit $N \rightarrow \infty$ and only the non-trivial commutators in the QRF $\mathsf A$ contribute. Therefore, as $N \rightarrow \infty$, we obtain
\begin{equation}
[\widetilde F_N(P), \widetilde F_N(K)] = \frac{m_{\mathsf a}^2}{m_{\mathsf A}^2} [P_{\mathsf A}, K_{\mathsf A}] = i \frac{m^2_{\mathsf a}}{m_{\mathsf A}},
\end{equation}
which is Eq.~\eqref{eq:galileicommrel}.

\subsection{SU(2)} \label{app:UR SU(2)}

Let $G=\SU(2)$ and $\ket \psi=\alpha\ket{0}+\beta\ket{1}$ be a qubit. Then, $\ket{\psi_N}=\vert \Omega_{\mathbf{u}}^{N/2}\rangle$ is a spin-$\frac{N}{2}$ coherent state aligned along a direction $\mathbf{u}$. Choosing a projector $P_{\mathbf{v}}=\frac{1}{2}(\id + \mathbf{v}\cdot\mathbf{\sigma})$, for a unit vector $\mathbf v$, we can write the relative frequency operator as
\bea
    F_N(P_{\mathbf{v}}) &= \frac{1}{2}\id^{\otimes N} + \frac{1}{N}\sum_i\frac{1}{2}(\mathbf{v}\cdot\mathbf{\sigma})^{(i)} \\
    & = \frac{1}{2}\id^{\otimes N} + \frac{1}{N}\hat{S}_{\mathbf{v}}.
\eea
where $\hat{S}_{\mathbf{v}}$ is the projection of the spin-$\frac{N}{2}$ operator along $\mathbf{v}$.
Writing $\mathbf{v}=\cos(\theta_{uv})\mathbf{u}+\sin(\theta_{uv})\mathbf{u}^\perp$ we calculate the action on $\ket{\psi_N}$ as
\be
    F_N(P_{\mathbf{v}})\ket{\psi_N}=F_N(P_{\mathbf{v}})\vert \Omega_{\mathbf{u}}^{N/2}\rangle=\frac{1}{2}(1+\cos(\theta_{uv}))\vert \Omega_{\mathbf{u}}^{N/2}\rangle + \frac{\sin(\theta_{uv})}{N}\hat{S}_{\mathbf{u}^\perp}\vert \Omega_{\mathbf{u}}^{\perp\,N/2}\rangle
\ee
and so $\langle F_N(P_{\mathbf{v}})\rangle_{\psi_N}=\bra{\psi}P_{\mathbf{v}}\ket{\psi}$.
So the relative frequency operator can be rewritten as a collective spin operator on the $N$-copies of the qubits.

Consequently, the relational relative frequency operator can be understood as measuring the relative orientation between the QRF $\mathsf A$ and the $N$-copies of the system $\mathsf a$, taken as a single collective spin:
\be
    \widetilde{F}_N(P_{\mathbf{v}}) = \frac{1}{2}\id_{\mathsf A \mathsf a^N} + \frac{1}{N}\widetilde{S}_{\mathbf{v}},
\ee
with
\be\label{eq:G-twirled J_v}
    \widetilde{S}_{\mathbf{v}} = \int_{\SU(2)}\rd g \,E(g)\otimes g\triangleright S_\mathbf{v} = \sum_{l,k}v_k\l(\int_{\SU(2)}\rd g \, E(g)\mathcal R_{lk}(g)\r) \otimes S_{\mathsf{a}^N,l},
\ee
where $\mathcal R_{lk}(g)$ are the elements of the  $\SU(2)$ rotation matrix $\mathcal R$. We consider a single $J$-sector of the regular representation of $\SU(2)$. By the Peter-Weyl theorem, the Hilbert space of $\mathsf A$ is then spanned by vectors of the form
\be \label{eq:regrep}
    \ket{g} = \sqrt{d_{J}}\sum_{m,n}D^{J}_{mn}(g)\,\ket{J,m}_L\otimes\ket{J,n}^*_R.
\ee
where $D^{J}(g)$ is the Wigner D-matrix of spin $J$, $d_{J}=2J+1$, the subscripts $L$ and $R$ denote respectively the left-regular and right-regular tensor factors of the reference frame. Eq.~\eqref{eq:regrep} is the vectorised version of a general $\SU(2)$ rotation matrix $U_{J}(g)=\sum_{m,n}D^{J}_{mn}(g)\ketbra{J,n}{J,m}$. Then,
\be
    E(g) = \ketbra{g}{g}, \qquad \text{ and } \qquad \int_{\SU(2)}E(g)\rd g = \id_\mathrm{\mathsf A}
\ee
Using the Wigner-Eckart theorem, \eqref{eq:G-twirled J_v} can be rewritten as
\be \label{eq:weckartresult}
    \widetilde{S}_\mathbf{v}=\alpha_{J}\sum_{kl}v_k {L}_{\mathsf{A},l}\otimes R_{\mathsf{A},k}\otimes S_{\mathsf{a}^N,l}.
\ee
where the operators $L$ and $R$ are the spin operators acting on the left- and right-regular tensor factors.
Integrating~\eqref{eq:weckartresult} with respect to $v$ and equating its trace with that of~\eqref{eq:G-twirled J_v} gives $\alpha_{J}=\frac{1}{J(J+1)}$ and so
\be\label{eq:RRF for SU(2)}
    \widetilde{F}_N(P_{\mathbf{v}}) = \frac{1}{2}\id_{\mathsf A \mathsf a^N} + \frac{1}{N}\frac{1}{J(J+1)}\sum_l L_{\mathsf{A},l}(\mathbf{v}\cdot{\mathbf{R}_{\sf A}})\otimes {S}_{\mathsf{a}^N,l}
\ee
which, inverting the place of the left-regular and right-regular representations in the tensor product can be rewritten as
\be\label{eq:RRF for SU(2) better order}
    \widetilde{F}_N(P_{\mathbf{v}})=\frac{1}{2}\id_{\mathsf A \mathsf a^N} + \frac{1}{N}\frac{1}{J(J+1)}\mathbf{v}\cdot {\mathbf{R}_{\mathsf A}}\otimes {\mathbf{L}_{\mathsf A}}\cdot {\mathbf{S}}_{\mathsf a^N}
\ee 
The expectation value of this operator in a state $\ket{\chi} = \ket{ e}\otimes\ket{\psi}^{\otimes N}$ (where $\ket{\psi}$ is a qubit along direction $\mathbf{u}$) gives (in the following, we omit Hilbert space indices for simplicity of notation)
\be\label{eq:expectation value RRF SU(2)}
    \langle\widetilde{F}_N(P_{\mathbf{v}})\rangle_\chi = \frac{1}{2} + \frac{1}{2}\frac{1}{J(J+1)}\bra{e}\mathbf{v}\cdot\mathbf{R}_\mathsf{A} \otimes \mathbf{u}\cdot\mathbf{L}_\mathsf{A}\ket{e}
\ee
The spectrum of this operator lies inside $\l[\frac{1}{2}\frac{1}{J+1},\frac{1}{2} + \frac{1}{2}\frac{J}{J+1}\r]$. In particular, for $J=1/2$, the probabilities lie in $[\frac{1}{3}, \frac{2}{3}]$. However, the action of $\widetilde{F}_N(P_{\mathbf{v}})$ on a state $\ket{\chi}$ still contains off diagonal terms which vanish only at the $N,J\to\infty$ limit.

The commutator between two RRFs for spin projections along two different directions $\mathbf{v}$ and $\mathbf{w}$ reads:
\be\label{eq:commutator SU(2) operator}
    [\widetilde{F}_N(P_\mathbf{v}),\widetilde{F}_N(P_\mathbf{w})] = \frac{1}{N^2J^2(J+1)^2}\,[\mathbf{v}\cdot\mathbf{R}_\mathsf{A},\mathbf{w}\cdot\mathbf{R}_\mathsf{A}]\otimes(\mathbf{L}_\mathsf{A}\cdot\hat{\mathbf{S}}_{\mathsf{a}^N})^2 = \frac{i}{N^2J^2(J+1)^2}(\mathbf{v}\times\mathbf{w})\cdot\mathbf{R}_\mathsf{A}\otimes(\mathbf{L}_\mathsf{A}\cdot\hat{\mathbf{S}}_{\mathsf{a}^N})^2.
\ee
Similarly
\be\label{eq:anticommutator SU(2) operator}
    \{\widetilde{F}_N(P_\mathbf{v}),\widetilde{F}_N(P_\mathbf{w})\} = \frac{2}{N^2J^2(J+1)^2}(\mathbf{v}\cdot\mathbf{w})\,\id_{R_\mathsf{A}}\otimes(\mathbf{L}_\mathsf{A}\cdot\hat{\mathbf{S}}_{\mathsf{a}^N})^2.
\ee
Consider a state of the form $\ket{\chi}= \ket{\phi}\ket{\psi}^{\otimes N}$, where $\ket \phi$ is an arbitrary pure state of the QRF and $\ket{\psi}$ is a qubit state along a direction $\mathbf{u}$. The expectation value of \eqref{eq:commutator SU(2) operator} in state $\ket{\chi}$ reads
\be\label{eq:commutator SU(2) operator infinite N}
   \bra{\chi}[\widetilde{F}_N(P_\mathbf{v}),\widetilde{F}_N(P_\mathbf{w})]\ket{\chi} = \frac{i}{8J^2(J+1)^2}\bra{\phi}(\mathbf{v}\times\mathbf{w})\cdot\mathbf{R}_\mathsf{A}\otimes(\mathbf{u}\cdot\mathbf{L}_\mathsf{A})^2\ket{  \phi}
\ee
and
\be\label{eq:anticommutator SU(2) operator infinite N}
\bra{\chi}\{\widetilde{F}_N(P_\mathbf{v}),\widetilde{F}_N(P_\mathbf{w})\}\ket{\chi} = \frac{1}{4J^2(J+1)^2} (\mathbf{v}\cdot\mathbf{w}) \bra \phi(\mathbf{u}\cdot\mathbf{L}_\mathsf{A})^2\ket{ \phi}.
\ee
Hence, using \eqref{eq:commutator SU(2) operator infinite N}, \eqref{eq:anticommutator SU(2) operator infinite N} and \eqref{eq:expectation value RRF SU(2)} we can compute the Robertson-Schr\"odinger uncertainty relation:
\bea\label{eq:UR SU(2) spin J}
    \Delta^2_\chi\widetilde{F}_N(P_\mathbf{v})\Delta^2_\chi\widetilde{F}_N(P_\mathbf{w})^2 \geq & \,\frac{1}{16}\frac{1}{J^4(J+1)^4}\bra{ \phi}(\mathbf{v}\times\mathbf{w})\cdot\mathbf{R}_\mathsf{A}\otimes(\mathbf{u}\cdot\mathbf{L}_\mathsf{A})^2\ket{\phi}^2 \\
    & + \l(\frac{1}{8}\frac{1}{J^2(J+1)^2}(\mathbf{v}\cdot\mathbf{w})\bra{ \phi}(\mathbf{u}\cdot\mathbf{L}_\mathsf{A})^2\ket{ \phi} - \frac{1}{4}\frac{1}{J(J+1)}\bra{\phi}(\mathbf{v}+\mathbf{w})\cdot\mathbf{R}_\mathsf{A}\otimes\mathbf{u}\cdot\mathbf{L}_\mathsf{A}\ket{ \phi}\r. \\
    & \l. \q\q\q  - \frac{1}{4}\frac{1}{J^2(J+1)^2}\bra{ \phi}\mathbf{v}\cdot\mathbf{R}_\mathsf{A}\otimes\mathbf{u}\cdot\mathbf{L}_\mathsf{A}\ket{\phi}\bra{ \phi}\mathbf{w}\cdot\mathbf{R}_\mathsf{A}\otimes\mathbf{u}\cdot\mathbf{L}_\mathsf{A}\ket{\phi} - \frac{1}{4} \r)^2 \\
    > & \, 0
\eea
For the case where $J=1/2$, the reference frame state takes the form $\ket{\phi}=\sum_{\alpha,\beta}M_{\alpha\beta}\ket{\beta}_R^*\ket{\alpha}_L$. In particular, for $\ket \phi = \ket e $, \eqref{eq:regrep} means that the QRF state is maximally entangled, thus leading to a vanishing contribution from the term in~\eqref{eq:commutator SU(2) operator infinite N}. However, the bound of~\eqref{eq:UR SU(2) spin J} is still not trivial. Specifically, using the $J=1/2$ operators $L_{\mathsf{A},i}=\frac{1}{2}\id\otimes\sigma_i$ and $R_{\mathsf{A},i}=\frac{1}{2}\sigma_i\otimes\id$ and so \eqref{eq:UR SU(2) spin J} we get
\be
   \Delta_{\chi, J=\frac{1}{2}}^2(\widetilde{F}_N(P_\mathbf{v}))\Delta_{\chi, J=\frac{1}{2}}^2(\widetilde{F}_N(P_\mathbf{w})) \geq
    \l(\frac{1}{18}\,\mathbf{v}\cdot\mathbf{w} - \frac{1}{12}\,\mathbf{v}\cdot\mathbf{u'} - \frac{1}{12}\,\mathbf{w}\cdot\mathbf{u'} - \frac{1}{36}(\mathbf{v}\cdot\mathbf{u'})(\mathbf{w}\cdot\mathbf{u'}) - \frac{1}{4}\r)^2,
\ee
where $\mathbf{u'}=(u_x,-u_y,u_z)$.
This bound is always greater or equal than $1/324$, which is achieved when $\mathbf{v}=\mathbf{w}=-\mathbf{u'}$.

\section{Uncertainty relations in the POVM construction}
\label{URPOVM}

\noindent We now show that the POVM construction introduced in
Section~\ref{rrfs} obeys an uncertainty
relation controlled by the same commutator as the Hermitian RRF first-moment
operators.

Let \(P_1\) and \(P_2\) be two projectors on the single-system Hilbert space
\(\mathcal H_{\mathsf a}\). For \(P\in\{P_1,P_2\}\), write the spectral
decomposition of the ordinary relative frequency operator as
\(F_N(P)=\sum_{\nu}\nu Q^{(P)}_{\nu}\), where \(\nu=k/N\) and
\(Q^{(P)}_{\nu}\) are the corresponding spectral projectors on
\(\mathcal H_{\mathsf{a}^N}\). The POVM construction relativises
these projectors according to
\(\widetilde Q^{(P)}_{\nu}:=\mathcal G[E(e)\otimes Q^{(P)}_{\nu}]\). The
family \(\{\widetilde Q^{(P)}_{\nu}\}_{\nu}\) is a POVM on
\(\mathcal H_{\mathsf A}\otimes\mathcal H_{\mathsf{a}^N}\), whose
first moment is precisely the relational relative frequency operator,
\(\widetilde F_N(P)=\sum_{\nu}\nu \widetilde Q^{(P)}_{\nu}
=\mathcal G[E(e)\otimes F_N(P)]\).

The second moment of the actually measured POVM outcomes is
\(\widetilde F^{(2)}_{N,\mathrm{POVM}}(P):=
\sum_{\nu}\nu^2\widetilde Q^{(P)}_{\nu}\). In general,
\(\widetilde F^{(2)}_{N,\mathrm{POVM}}(P)\neq \widetilde F_N(P)^2\), because
the effects \( \widetilde Q^{(P)}_{\nu} \) need not be mutually orthogonal
projectors. Hence, for a state \(\rho\), the operational variance of the POVM
outcome distribution is
\be
    \Delta^2_{\mathrm{POVM},\rho}(\widetilde F_N(P))
    :=
    \mathrm{Tr}\!\left[
        \rho\,\widetilde F^{(2)}_{N,\mathrm{POVM}}(P)
    \right]
    -
    \mathrm{Tr}\!\left[
        \rho\,\widetilde F_N(P)
    \right]^2 .
\ee

To compare this with the variance of the Hermitian first-moment operator
\(\widetilde F_N(P)\), define the intrinsic noise operator
\be
    \mathcal N_N(P)
    :=
    \widetilde F^{(2)}_{N,\mathrm{POVM}}(P)
    -
    \widetilde F_N(P)^2 .
\ee
This operator is positive. Indeed, by Naimark dilation the POVM
\(\{\widetilde Q^{(P)}_{\nu}\}_{\nu}\) can be represented as
\(\widetilde Q^{(P)}_{\nu}=V_P^\dagger \widehat Q^{(P)}_{\nu}V_P\), where
\(V_P\) is an isometry and \(\{\widehat Q^{(P)}_{\nu}\}_{\nu}\) is a PVM on a
larger Hilbert space. If
\(\widehat F_N(P):=\sum_{\nu}\nu\widehat Q^{(P)}_{\nu}\), then
\(\widetilde F_N(P)=V_P^\dagger\widehat F_N(P)V_P\) and
\(\widetilde F^{(2)}_{N,\mathrm{POVM}}(P)
=V_P^\dagger\widehat F_N(P)^2V_P\). Therefore
\(\mathcal N_N(P)=
V_P^\dagger\widehat F_N(P)(\id-V_PV_P^\dagger)\widehat F_N(P)V_P\geq0\), since \(V_PV_P^\dagger\) is an orthogonal projector, and hence
\(\id-V_PV_P^\dagger\) is positive.

It follows that the POVM variance decomposes as~\cite{massar2007uncertainty} 
\begin{equation}
    \Delta^2_{\mathrm{POVM},\rho}(\widetilde F_N(P))
    =
    \Delta^2_{\rho}\!\left(\widetilde F_N(P)\right)
    +
    \mathrm{Tr}\!\left[
        \rho\,\mathcal N_N(P)
    \right],
\end{equation}
where
\(\Delta^2_{\rho}(\widetilde F_N(P))
:=\mathrm{Tr}[\rho\,\widetilde F_N(P)^2]
-\mathrm{Tr}[\rho\,\widetilde F_N(P)]^2\). Since
\(\mathcal N_N(P)\geq0\), the POVM construction can only increase the
variance relative to the corresponding first-moment operator.

Applying the uncertainty relation~\eqref{eq:uncertaintymain} to the Hermitian first-moment
operators \(\widetilde F_N(P_1)\) and \(\widetilde F_N(P_2)\), and using the
positivity of the intrinsic noise terms, we obtain
\be
    \Delta^2_{\mathrm{POVM},\rho}(\widetilde F_N(P_1))
    \Delta^2_{\mathrm{POVM},\rho}(\widetilde F_N(P_2))
    \geq
    \frac14
    \left\vert 
        \mathrm{Tr}\!\left[
            \rho\,
            [\widetilde F_N(P_1),\widetilde F_N(P_2)]
        \right]
    \right\vert ^2 .
\ee
Thus the POVM construction inherits the incompatibility of the relational
relative frequency first moments. The difference from the PVM construction is
that the actually measured POVM variances contain the additional positive
contributions \(\mathrm{Tr}[\rho\,\mathcal N_N(P_1)]\) and
\(\mathrm{Tr}[\rho\,\mathcal N_N(P_2)]\), which quantify the intrinsic unsharpness of
the relativised frequency POVMs.

\section{Relativised frequency POVMs admit a local hidden-variable description}
\label{sec:povm-outcome-statistics-classical-simulation}
\label{app:povmhiddenvariable}
\noindent We now discuss the statistics of the relativised frequency POVM
introduced in Sec.~\ref{rrfs}. The key point is that this POVM is itself a
classical mixture of ordinary sharp-reference-frame frequency measurements. A
group element $g$ is sampled according to the probability distribution arising from the POVM associated with the finite reference frame, and,
conditioned on $g$, one performs the ordinary sharp-frame frequency
measurement rotated by $g$. In the large-$N$ relative frequency limit, this
leads to a local hidden-variable description of the Bell-type statistics of the
POVM construction. In the following we prove these results assuming a separable state between reference frames and systems. 

Let the finite reference frame be prepared in a state $\rho_\mathsf{A}$. The
covariant POVM $\{E(g)\}_{g\in G}$, normalized as
$\int_G d g\,E(g)=\id_\mathsf{A}$, defines the probability density
$w_\mathsf{A}(g):=\operatorname{Tr}(E(g)\rho_\mathsf{A})$, with $w_\mathsf{A}(g)\geq0$ and
$\int_G d g\,w_\mathsf{A}(g)=1$. For fixed $g$, let $Q_\nu^{(P_g)}$ denote the
ordinary sharp-frame frequency POVM for the two-outcome effect 
$P_g=g\triangleright P$, with outcome $\nu=k/N$. Multi-outcome effects can be treated analogously. The corresponding relativised
frequency POVM on the joint system consisting of the reference frame $\mathsf{A}$ and the
$N$ systems $\mathsf{a}$ is
\begin{equation}
    \widetilde Q_\nu^{(P)}
    =
    \int_G d g\,
    E(g)\otimes Q_\nu^{(P_g)}.
    \label{eq:relativised-frequency-povm-effect}
\end{equation}
This is a POVM because each effect is positive and
$\sum_k\widetilde Q_\nu^{(P)}=\id_\mathsf{A}\otimes\id_{\mathsf{a}^N}$.

For the product state $\rho_\mathsf{A}\otimes\rho_{\mathsf{a}^N}$, the probability of obtaining
$K=k$ is therefore
\begin{align}
    \mathrm{Prob}_N(K=k)
    &=
    \operatorname{Tr}\l(
        \widetilde Q_\nu^{(P)}\,
        (\rho_\mathsf{A}\otimes\rho_{\mathsf{a}^N})\r) \nonumber \\
    &=
    \int_G d g\,
    w_\mathsf{A}(g)\,
    \operatorname{Tr}_{\mathsf{a}^N}\!\left(
         Q_\nu^{(P_g)}\rho_{\mathsf{a}^N}
    \right).
    \label{eq:povm-classical-randomisation}
\end{align} Equivalently, after fixing the reference-frame state $\rho_\mathsf{A}$, the induced
POVM on the $N$ systems alone \ecr{is~\cite{loveridge2018symmetry}}
\begin{equation}
        \widetilde Q_{\nu\vert \rho_\mathsf{A}}^{(P)}
    =
    \int_G d g\,
    w_\mathsf{A}(g)\, Q_\nu^{(P_g)}.\label{mixturePOVM}
\end{equation}
Thus the relativised POVM statistics are exactly those of a classical mixture
of ordinary sharp-frame frequency measurements: a classical orientation $g$
is sampled from $w_\mathsf{A}(g)$, and then the sharp-frame frequency measurement
rotated by $g$ is performed.

For a product preparation $\rho_\mathsf{a}^{\otimes N}$, and with
$p(g)=\operatorname{Tr}\l(P_g\rho_\mathsf{a}\r)$, Eq.~\eqref{eq:povm-classical-randomisation}
becomes the binomial mixture
\begin{equation}
    \mathrm{Prob}_N(K=k)
    =
    \binom Nk
    \int_G d g\,
    w_\mathsf{A}(g)\,
    p(g)^k\l(1-p(g)\r)^{N-k}.
    \label{eq:binomial-mixture-relativised-povm}
\end{equation}
Conditioned on $g$, this is the ordinary frequency statistics of $N$
identical trials. Hence, by Theorem \ref{theorem: Finkelstein-Hartle-Herbut}, the relative frequency $f_N=K/N$ conditioned on $g$ converges to the probability given by the Born rule $p(g)$. Equivalently,
\begin{equation}
    \mathrm{Prob}_\infty(f\in df)
    :=
    \lim_{N\to\infty}\mathrm{Prob}_N(f_N\in df)
    =
    \int_G d g\,
    w_\mathsf{A}(g)\,
    \delta\!\left(df-p(g)\right).
    \label{eq:single-system-povm-frequency-limit}
\end{equation}
The limiting distribution $f$ is therefore a classical distribution over
ideal-frame Born probabilities.

We now spell out the bipartite analogue of
Eq.~\eqref{eq:binomial-mixture-relativised-povm}. This is the relevant object
for evaluating Bell-inequality expressions, because one is interested in the joint probability
that Alice obtains a relative frequency $K_\mathsf{A}/N$ and Bob obtains a relative
frequency $K_\mathsf{B}/N$.

Alice and Bob hold finite reference frames $\mathsf{A}$ and $\mathsf{B}$, possibly in an
entangled state $\rho_\mathsf{AB}$. Their local orientation POVMs define the joint
probability density
\be
    w_\mathsf{AB}(g_\mathsf{A},g_\mathsf{B})
    :=
    \operatorname{Tr}_\mathsf{AB}\!\l(
        \bigl(E_\mathsf{A}(g_\mathsf{A})\otimes E_\mathsf{B}(g_\mathsf{B})\bigr)\rho_\mathsf{AB}
    \r),
\ee
with $w_\mathsf{AB}(g_\mathsf{A},g_\mathsf{B})\geq0$ and
$\int_G d g_\mathsf{A}\int_G d g_\mathsf{B}\,w_\mathsf{AB}(g_\mathsf{A},g_\mathsf{B})=1$. Thus the POVM
construction treats the finite reference frames as a classical random
orientation variable $\lambda=(g_\mathsf{A},g_\mathsf{B})$ distributed according to
$w_\mathsf{AB}$.

Assume first that the $2N$ systems are prepared in $N$ identical pairs,
$\rho_{\mathsf{a}^N \mathsf{b}^N}=\rho_\mathsf{ab}^{\otimes N}$, where the elementary pair state
$\rho_\mathsf{ab}$ may be entangled. For fixed frame orientations $(g_\mathsf{A},g_\mathsf{B})$ and
settings $(x,y)$, define
\be
    q_{\alpha\beta}^{xy}(g_\mathsf{A},g_\mathsf{B})
    :=
    \operatorname{Tr}\!\l(
        \bigl(P^{(\alpha)}_{x,g_\mathsf{A}}\otimes Q^{(\beta)}_{y,g_\mathsf{B}}\bigr)
        \rho_\mathsf{ab}
    \r),
    \qquad
    \alpha,\beta\in\{0,1\},
\ee
where $P^{(1)}_{x,g_\mathsf{A}}=g_\mathsf{A}\triangleright P_x$,
$P^{(0)}_{x,g_\mathsf{A}}=\id-g_\mathsf{A}\triangleright P_x$, and analogously
$Q^{(1)}_{y,g_\mathsf{B}}=g_\mathsf{B}\triangleright Q_y$,
$Q^{(0)}_{y,g_\mathsf{B}}=\id-g_\mathsf{B}\triangleright Q_y$. If $\rho_\mathsf{ab}$ is
entangled, the four-outcome distribution $q_{\alpha\beta}^{xy}$ need not
factorise.

Let $N_{00},N_{10},N_{01},N_{11}$ denote the four outcome counts in $N$
repetitions, with $N_{00}+N_{10}+N_{01}+N_{11}=N$. Alice's and Bob's numbers
of $1$-outcomes are $K_\mathsf{A}=N_{10}+N_{11}$ and $K_\mathsf{B}=N_{01}+N_{11}$.
Conditioned on $(g_\mathsf{A},g_\mathsf{B},x,y)$, the four counts are multinomially
distributed. Hence the probability that Alice obtains $K_\mathsf{A}=k_\mathsf{A}$ and Bob
obtains $K_\mathsf{B}=k_\mathsf{B}$ is
\be
    C_N(k_\mathsf{A},k_\mathsf{B}\vert x,y,g_\mathsf{A},g_\mathsf{B}) =
    \sum_{\ell=\ell_{\min}}^{\ell_{\max}}
    \frac{N!}{
        n_{00}!\,n_{10}!\,n_{01}!\,n_{11}!
    }
    (q_{00}^{xy})^{n_{00}}
    (q_{10}^{xy})^{n_{10}}
    (q_{01}^{xy})^{n_{01}}
    (q_{11}^{xy})^{n_{11}},
\ee
where $n_{11}=\ell$, $n_{10}=k_\mathsf{A}-\ell$, $n_{01}=k_\mathsf{B}-\ell$,
$n_{00}=N-k_\mathsf{A}-k_\mathsf{B}+\ell$, and
$\ell_{\min}=\max\{0,k_\mathsf{A}+k_\mathsf{B}-N\}$,
$\ell_{\max}=\min\{k_\mathsf{A},k_\mathsf{B}\}$. The bipartite analogue of
Eq.~\eqref{eq:binomial-mixture-relativised-povm} is therefore
\begin{equation}
    \mathrm{Prob}_N(K_\mathsf{A}=k_\mathsf{A},K_\mathsf{B}=k_\mathsf{B}\vert x,y)  =
    \int_G d g_\mathsf{A}
    \int_G d g_\mathsf{B}\,
    w_\mathsf{AB}(g_\mathsf{A},g_\mathsf{B})\,
    C_N(k_\mathsf{A},k_\mathsf{B}\vert x,y,g_\mathsf{A},g_\mathsf{B}).
\label{eq:bipartite-frequency-mixture}
\end{equation}
This is the exact finite-$N$ quantum probability of the relativised POVM,
written as an average over the classical frame orientations. 

We now take the relative frequency limit $N\to\infty$. Write
$f_\mathsf{A}=K_\mathsf{A}/N$ and $f_\mathsf{B}=K_\mathsf{B}/N$. For fixed $(g_\mathsf{A},g_\mathsf{B},x,y)$, the multinomial
law of large numbers gives
$N_{\alpha\beta}/N\to q_{\alpha\beta}^{xy}(g_\mathsf{A},g_\mathsf{B})$
jointly for all four outcomes. Therefore $f_\mathsf{A}$ and $f_\mathsf{B}$ converge to
the local Born probabilities
\be
\begin{aligned}
    p_\mathsf{A}(x,g_\mathsf{A})
    &:=
    q_{10}^{xy}(g_\mathsf{A},g_\mathsf{B})+q_{11}^{xy}(g_\mathsf{A},g_\mathsf{B})
     =
    \operatorname{Tr}\!\l(
        \bigl((g_\mathsf{A}\triangleright P_x)\otimes\id\bigr)\rho_\mathsf{ab}
    \r),
    \\
    p_\mathsf{B}(y,g_\mathsf{B})
    &:=
    q_{01}^{xy}(g_\mathsf{A},g_\mathsf{B})+q_{11}^{xy}(g_\mathsf{A},g_\mathsf{B})
     =
    \operatorname{Tr}\!\l(
        \bigl(\id\otimes(g_\mathsf{B}\triangleright Q_y)\bigr)\rho_\mathsf{ab}
    \r).
\end{aligned}
\ee
Thus, conditioned on $(g_\mathsf{A},g_\mathsf{B})$, the limiting joint distribution of the two
relative frequencies is the product of two delta measures. Substituting this
limit into Eq.~\eqref{eq:bipartite-frequency-mixture} gives
\begin{equation}
\begin{aligned}
    \mathrm{Prob}_\infty(df_\mathsf{A},df_\mathsf{B}\vert x,y)  & :=
    \lim_{N\to\infty}
    \mathrm{Prob}_N(df_\mathsf{A},df_\mathsf{B}\vert x,y)  \\
    &\quad =
    \int_G d g_\mathsf{A}
    \int_G d g_\mathsf{B}\,
    w_\mathsf{AB}(g_\mathsf{A},g_\mathsf{B})\,
    \delta\!\left(df_\mathsf{A}-p_\mathsf{A}(x,g_\mathsf{A})\right)
    \delta\!\left(df_\mathsf{B}-p_\mathsf{B}(y,g_\mathsf{B})\right).
\end{aligned}
\label{eq:rrf-povm-lhv-frequency-limit}
\end{equation}
Equation~\eqref{eq:rrf-povm-lhv-frequency-limit} is therefore the
$N\to\infty$ limit of the quantum probabilities of the relativised POVM for
the joint event that Alice obtains a relative frequency $f_\mathsf{A}$ and Bob obtains
a relative frequency $f_\mathsf{B}$.

The crucial observation is that, conditioned on the classical variable
$(g_\mathsf{A},g_\mathsf{B})$, Alice's limiting relative frequency outcome depends only on her
local setting $x$, and Bob's limiting relative frequency outcome depends only
on his local setting $y$. Equation~\eqref{eq:rrf-povm-lhv-frequency-limit}
therefore has precisely the form required by Bell's local-causality condition:
all correlations between Alice's and Bob's limiting frequency outcomes arise
only through the common hidden variable $(g_\mathsf{A},g_\mathsf{B})$ describing the finite
reference frames. Thus, in the infinite-repetition limit, the relativised
frequency POVM statistics admit a local hidden-variable model.

The same conclusion holds for exchangeable preparations that admit a de
Finetti representation, Eq.~\eqref{eq:de Finetti}, of the form
$\rho_{\mathsf{a}^N \mathsf{b}^N} =\int d \sigma\,\sigma_\mathsf{ab}^{\otimes N}$. For each value of the
classical parameter $\sigma$, the preceding argument applies with the
elementary pair state $\rho_\mathsf{ab}$ replaced by $\sigma_\mathsf{ab}$, and with the
expanded hidden variable $\lambda=(\sigma,g_\mathsf{A},g_\mathsf{B})$.

\section{Bell inequalities}
\subsection{For relational relative frequencies}
\label{appendix:bell-inequality}
\label{app:relationalrelfrequancies}

In Appendix \ref{app:UR SU(2)}, the RRF operators were shown in general to not commute for the case of an $\SU(2)$ reference frame for any $N$. We next show that, the incompatibility of local RRF observables is sufficient to obtain Bell-inequality violations for reference frames associated with $\SU(2)$. Consider again Alice, Bob, and their respective reference frames and relational relative frequency operators. Let us define the $\pm1$-valued observable
\begin{equation}\label{eq: dichotomic observable general}
    \mathcal{A}_N(P_\mathbf{v}) := \sgn(\widetilde{F}_N(P_\mathbf{v}) - \frac{1}{2}\id) = \sgn{\l((\mathbf{v}\cdot \mathbf{R}_{\mathsf A})\otimes{\mathbf{L}_{\mathsf A}} \cdot {\mathbf{S}}_{\mathsf a^N}\r)},
\end{equation}
which specifies whether the observed relational relative frequency is smaller than $\frac{1}{2}$ or not. The dichotomic observable $\mathcal{A}_N(P_\mathbf{w})$ on Bob's side is similarly defined.

Consider the case $J=1/2$, where the expression~\eqref{eq: dichotomic observable general} reduces to
\bea
    \mathcal{A}_N^{J=\frac{1}{2}}(P_\mathbf{v}) & = \sgn{\l((\mathbf{v}\cdot \mathbf{\sigma}_{\mathsf A}^R) \otimes\mathbf{\sigma}_{\mathsf A}^L \cdot \mathbf{S}_{\mathsf a^N}\r)} \\
    \mathcal{B}_N^{J=\frac{1}{2}}(P_\mathbf{w}) & = \sgn{\l((\mathbf{w}\cdot \mathbf{\sigma}_{\mathsf B}^R) \otimes\mathbf{\sigma}_{\mathsf B}^L \cdot \mathbf{S}_{\mathsf b^N}\r)}
\eea
The total state for our scenario is the following:
\begin{equation}
    \ket{\Psi} = \ket{\Phi}_{\mathsf{AB}}\otimes\ket{\psi}_{\mathsf{a}^N}\otimes\ket{\psi}_{\mathsf{b}^N},
\end{equation}
where $\ket{\Phi}_\mathsf{AB}$ is an entangled state between Alice's and Bob's reference frames of the form 
\begin{equation}
    \ket{\Phi}_\mathsf{AB}=\ket{\psi^-}_{R_\mathsf{AB}}\ket{\mathbf{\omega},+}_{L_\mathsf{A}}\ket{\mathbf{\omega},+}_{L_\mathsf{B}},
\end{equation} 
and $\ket{\psi}_{\mathsf a}$ is a qubit aligned along direction $\mathbf u$. Here the left sectors are in a product state but the right ones are prepared in the singlet state $\ket{\psi^-}$. The correlation function reads
\bea
    E(x,y) & = \bra{\Phi}\sgn{\l(\mathbf{v}_x\cdot \mathbf{\sigma}_{\mathsf A}^R \otimes\mathbf{u}\cdot \mathbf{\sigma}_{\mathsf A}^L\r)}\otimes\sgn{\l(\mathbf{w}_y\cdot \mathbf{\sigma}_{\mathsf B}^R \otimes\mathbf{u}\cdot \mathbf{\sigma}_{\mathsf B}^L\r)}\ket{\Phi}_\mathsf{AB} \\
    & = \bra{\psi^-}\sgn{\l(\mathbf{v}_x\cdot \mathbf{\sigma}_{\mathsf A}^R\r)}\otimes\sgn{\l(\mathbf{w}_y\cdot \mathbf{\sigma}_{\mathsf B}^R\r)}\ket{\psi^-}_{R_\mathsf{AB}} \\
    & = -\mathrm{v}_x\cdot\mathrm{w}_y,
\eea
This is independent of $N$ and so in particular holds for $N\to\infty$.
Thus the Bell experiment implemented by sign-binned RRFs is exactly the usual qubit singlet
experiment on the left/right sectors. One can then choose appropriate settings $\mathrm{v}_x$ and $\mathrm{w}_y$ with $x,y\in\{0,1\}$ in order to obtain the Tsirelson violation bound $2\sqrt{2}$.
\subsection{For relational expectation values}
\label{app:proofofbell}

Let $\mathsf{A}$ be a QRF for the Heisenberg-Weyl group -- a particle on a line -- with dimensionless position and momentum operators $X,P$, $[X,P]=i$, and let $\mathsf a_1,\dots,\mathsf a_N$ be $N$ systems with mutually isomorphic Hilbert spaces, which we identify with a common $\cH_{\mathsf{a}}$. The joint Hilbert space is $\mathcal H_{\mathsf A}\otimes\mathcal H_{\mathsf a^N}$, where $\mathcal H_{ \mathsf a^N} = \mathcal H_{\mathsf a}^{\otimes N}$.
We know from Appendix~\ref{app:uncertaintyrelations} that the relational expectation value operator on $\mathcal H_{\mathsf A}\otimes\mathcal H_{\mathsf a^N}$ is given by
\begin{align}
\widetilde F_N(X) =& \int \frac{\mathrm{d}\alpha}{\pi} \, \ketbra{\alpha}{\alpha} \otimes D_\alpha^{\otimes N} F_N (X)  D_\alpha^{\dagger \otimes N} = \id \otimes F_N(X) - X \otimes \id\\
\widetilde F_N(P) =& \int \frac{\mathrm{d}\alpha}{\pi} \, \ketbra{\alpha}{\alpha} \otimes D_\alpha^{\otimes N} F_N (P)  D_\alpha^{\dagger \otimes N} = \id \otimes F_N(P) - P \otimes \id.
\end{align}

Because both $F_N$ and $\widetilde F_N$ are linear in their operator argument, for an arbitrary quadrature $X_\vartheta\;=\;\cos\vartheta  X + \sin\vartheta P$
we have
\begin{equation}
\widetilde F_N(X_\vartheta)\; = \; \int \frac{\mathrm{d}\alpha}{\pi} \, \ketbra{\alpha}{\alpha} \otimes D_\alpha^{\otimes N} F_N (X_\vartheta)  D_\alpha^{\dagger \otimes N} \;=\; \id \otimes F_N(X_\vartheta)\; - \; X_\vartheta \otimes \id,
\label{eq:relquad}
\end{equation}
where
\begin{equation}
F_N(X_\vartheta) = \frac{1}{N} \sum _k X_{\vartheta, k}, 
\end{equation}
with $X_{\vartheta, k}$ acting nontrivially only on $\mathcal H_{\mathsf a_k}$.

Let $\ket{q}_\vartheta$ denote the ($\delta$-normalised) eigenstate of $X_\vartheta$ with eigenvalue $q$. The state $\ket{q}_\vartheta$ can be obtained by rotating the position eigenstate $\ket{q}$, $X \ket{q} = q \ket q$, in phase-space by the angle $\vartheta$: $\ket{q}_\vartheta=e^{-i\vartheta\hat n}\ket{q}$, where $\hat n$ is the number operator. In this way, the vectors $\ket{q}_\vartheta$ satisfy 
\begin{equation}
\bra{q}q'\rangle_{\vartheta}=\delta(q-q'),\qquad \int \mathrm d q\,\ket{q}\!\bra{q}_\vartheta = \id .
\end{equation}
On $\mathsf a^N$, define the product
\begin{equation}
\ket{\vec q \, }_\vartheta\;:=\;\bigotimes_{k=1}^{N}\ket{q_k}_\vartheta,\qquad
\vec q  \,=(q_1,\dots,q_N)\in\mathbb R^N,
\end{equation}
which satisfies
\begin{equation}
X_{\vartheta,k}\ket{\vec q  \,}_\vartheta=q_k\ket{\vec q  \,}_\vartheta,\qquad
F_N(X_\vartheta)\ket{\vec q  \,}_\vartheta= \bar q \ket{\vec q  \,}_\vartheta,
\end{equation}
where $\bar q := (1/N)\sum_{k=1}^{N}q_k$. Both terms in Eq.~\eqref{eq:relquad} are then diagonal on the product state $\ket{q}_\vartheta\otimes\ket{\vec q  \,}_\vartheta$, and equation \eqref{eq:eigenvalue} holds
\begin{equation*}
    \widetilde F_N(X_\vartheta)\,\bigl(\ket{q}_\vartheta\otimes\ket{\vec q  \,}_\vartheta\bigr)\;=\;
    (\bar q-q)\,\ket{q}_\vartheta\otimes\ket{\vec q  \,}_\vartheta .
\end{equation*}
The eigenvalue is the mean quadrature (in direction $\vartheta$) of the $N$-copies of system $\mathsf{a}$ relative to the QRF $\mathsf{A}$. Note that the operator from in Eq.~\eqref{eq:eigenvalue} is different from the mean relative quadrature of a single system, averaged over the $N$ systems. The eigenstate of that operator is one in which $\mathsf{A}$ and every system $\mathsf{a}$ are simultaneously in sharp $X_\vartheta$-eigenstates, so every individual relative quadrature $q_k-q$ is well defined, and so is any linear functional of them, in particular the mean.

We then have the spectral resolution
\begin{equation}
\widetilde F_N(X_\vartheta)\;=\;
\int_{\mathbb R}\!dq\!\int_{\mathbb R^N}\!d^N\vec q  \,\;
\Big(\bar q- q \Big)\;
\ket{q}\!\bra{q}_\vartheta\otimes\ket{\vec q  \,}\!\bra{\vec q  \,}_\vartheta.
\label{eq:spectral}
\end{equation}
The spectrum is $\mathbb R$, continuous and infinitely degenerate (the degeneracy reflects the freedom to redistribute a fixed value of the mean relative quadrature among the $N+1$ individual quadratures).

From the spectral resolution~\eqref{eq:spectral} we define the bounded, dichotomic ``sign'' observable
\begin{equation}
\sgn\!\bigl(\widetilde F_N(X_\vartheta)\bigr)\;:=\;
\int dq\,d^N\vec q  \,\;\sgn\!\Big(\bar q-q\Big)\;
\ket{q}\!\bra{q}_\vartheta\otimes\ket{\vec q  \,}\!\bra{\vec q  \,}_\vartheta,
\label{eq:sgn}
\end{equation}
where the sign function $\mathrm{sgn}$ is defined by $\mathrm{sgn}(x) = 1$ if $x > 0$ and $\mathrm{sgn}(x) = -1$ otherwise. Note that this implies that the observable $\sgn\!\bigl(\widetilde F_N(X_\vartheta)\bigr)$ has eigenvalues $\pm 1$. 

Now consider an analogous construction on Bob's side -- a QRF $\mathsf{B}$ with its own $N$ copies of system $\mathsf b$, $\mathsf b_1,\dots,\mathsf b_N$ -- and let Alice and Bob each pick one of two measurement angles, $\vartheta_0,\vartheta_1$ for Alice and $\varphi_0,\varphi_1$ for Bob. Define
\begin{equation}
E(\vartheta_{ a},\varphi_{ b})\;:=\;
\sgn\!\bigl(\widetilde F_N(X_{\vartheta_{\mathsf a}})\bigr)\,\otimes\,
\sgn\!\bigl(\widetilde F_N(X_{\varphi_{ b}})\bigr),
\end{equation}
and the CHSH--type correlator
\begin{equation}\label{eq:CHSHcorrelator}
 C\;:=\;\sum_{ a, b\,=\,0}^{1}(-1)^{ab}\,E(\vartheta_{a},\varphi_{ b}).
\end{equation}
Consider the state
\begin{equation}
\ket{\Psi}\;=\;\ket{\Phi}_{\mathsf A\mathsf B}\otimes\ket{\gamma}_{\mathsf a}^{\otimes N}\otimes\ket{\gamma}_{\mathsf b}^{\otimes N},
\end{equation}
where $\ket{\Phi}_{\mathsf  A \mathsf B}$ is an arbitrary bipartite state on $\mathcal H_{\mathsf A}\otimes\mathcal H_{\mathsf B}$ and each copy of $\mathsf{a}$ and $\mathsf{b}$ is in the state $\ket{\gamma}$ (taken the same on Alice's and Bob's sides for simplicity). Contracting Eq.~\eqref{eq:sgn} for Alice's and Bob's sign operators with the states on $\mathsf A$, $\mathsf B$, $\mathsf a^N,$ and $ \mathsf b^N$ gives
\begin{equation}
    \bra{\Psi}E(\vartheta_a,\varphi_b)\ket{\Psi}\;=\;
    \int dq\,dq'\;\mathcal S_N^{(\vartheta_a)}(q)\,\mathcal S_N^{(\varphi_b)}(q')\;
    \bra{\Phi}\bigl(\ket{q}\!\bra{q}_{\vartheta_a}\otimes\ket{q'}\!\bra{q'}_{\varphi_b}\bigr)\ket{\Phi}_{\mathsf  A \mathsf B},
\end{equation}
where 
\begin{equation}
S_N^{(\vartheta)}(q)\;:=\;\int_{\mathbb R^N}\!d^N\vec q\;
\prod_{k=1}^{N}\!\Big\vert  \bra{\gamma}q_k\rangle_\vartheta \Big\vert ^2\;
\sgn\Big(\bar q-q\Big).
\label{eq:Sdef}
\end{equation}
The squared product in Eq.~\eqref{eq:Sdef} is the joint probability density of $N$ i.i.d. systems, where the $k$-th one is drawn from the probability distribution $p_\vartheta(q_k):=\vert \bra{\gamma}q_k\rangle_{\vartheta}\vert ^2$.

By the strong law of large numbers,
\begin{equation}
\bar q = \frac{1}{N}\sum_{k=1}^{N}q_k\;\xrightarrow[\,N\to\infty] \;\mu_\vartheta\;:=\;\int dq\,q\,p_\vartheta(q)
\;=\;\bra{\gamma}X_\vartheta\ket{\gamma}
\end{equation}
Hence,
\begin{equation}
\mathcal S_N^{(\vartheta)}(q)\;\xrightarrow[\,N\to\infty\,]\;\sgn(\mu_\vartheta-q).
\end{equation}
In the case where $\ket{\gamma}$ is the vacuum state one has $\mu_\vartheta=0$ for every $\vartheta$ and
\begin{equation}
\sgn\!\Big(\frac{1}{N}\sum_k q_k-q\Big)\;\xrightarrow[\,N\to\infty\,]\;\sgn(-q)\;=\;-\sgn(q).
\end{equation}
At the operator level, for large $N$ restricted to states of the assumed product form, the relational sign~\eqref{eq:sgn} reduces to (minus) the absolute sign on the QRF:
\begin{equation}
\sgn\!\bigl(\widetilde F_N(X_\vartheta)\bigr) + \,\sgn(X_\vartheta)_{\mathsf A}\otimes\id_{ \mathsf a^N}\;\xrightarrow[\,N \to \infty \,]\; 0.
\end{equation}
Substituting this on both sides of $E(\vartheta_{\mathsf a},\varphi_{\mathsf b})$, the two minus signs multiply to $+1$ and we obtain
\begin{equation}
\lim_{N\to\infty}\bra{\Psi}E(\vartheta_{a},\varphi_{ b})\ket{\Psi}\;=\;
\bra{\Phi}\sgn(X_{\vartheta_{a}})_{\mathsf A}\otimes\sgn(X_{\varphi_{ b}})_{\mathsf B}\ket{\Phi}_{\mathsf  A \mathsf B}.
\label{eq:reduction}
\end{equation}
Summing over the four settings, we have
\begin{equation}
\lim_{N\to\infty}\langle  C \rangle_\Psi\;=\;\sum_{ a, b=0}^{1}(-1)^{ab}\,
\bra{\Phi}\sgn(X_{\vartheta_{ a}})_{\mathsf A}\otimes\sgn(X_{\varphi_{ b}})_{\mathsf B}\ket{\Phi}_{\mathsf  A \mathsf B}.
\end{equation}
We therefore see that the CHSH--type correlator~\eqref{eq:CHSHcorrelator}, built from REVs, reduces, in the large-$N$ limit, to the standard continuous--variable CHSH expression for quadrature--dichotomized observables on the bipartite state $\ket{\Phi}_{\mathsf  A \mathsf B}$.

It is well known that there exist angles $\vartheta_0$, $\vartheta_1$, $\varphi_0$ and $\varphi_1$ and states $\ket{\Phi}_{\mathsf  A \mathsf B}$ such that 
\begin{equation}
    \langle C \rangle_\Phi> 2.
\end{equation}
For instance, in Ref.~\cite{garcia2004proposal} it is shown that for $\vartheta_0 =0$, $\vartheta_1 =\pi/2$, $\varphi_0 = -\pi/4$ and $\varphi_1 = \pi/4$ a violation of $C \approx 2.048$ can be achieved for the ``degaussified'' two-mode squeezed vacuum state 
\begin{equation}\label{eq:two-mode squeezed state}
\ket{\Phi} \propto \sum_{k =0}^\infty (k + 1) r^k\, \ket{k}_{\mathsf A}\otimes \ket{k}_{\mathsf B}
\end{equation}
where $\ket{k}$ is the Fock basis and $r \approx 0.572$. 

We thus conclude that REVs lead to the violation of Bell inequality, and therefore their irreducible uncertainty in the large $N$ limit is inconsistent with every local hidden variable model.    

\section{Quantum-optical proposal for a PVM}
\label{app:pvm-quantum-optics}

\noindent We now describe a continuous-variable quantum-optical implementation of the
PVM construction introduced in Sec.~\ref{rrfs}. In that construction one
measures the spectral projectors of the Hermitian relational observable,
as in Eq.~\eqref{eq:rrf-pvm-spectral-decomposition}. The present proposal
specializes this idea to the Heisenberg--Weyl setting discussed in the
main text and in Appendix~\ref{app:uncertaintyrelations}, where the
relevant relational observables are mean quadratures of a block of \(N\)
systems measured relative to a finite quantum reference frame. The
experimental ingredients are standard in pulsed continuous-variable optics:
balanced homodyne readout of optical quadratures~\cite{LvovskyRaymer2009,HansenEtAl2001,AppelEtAl2007},
coherent linear-optical transformations on many modes~\cite{ReckZeilingerBernsteinBertani1994,ClementsEtAl2016},
and temporal-mode control of pulsed light~\cite{BrechtReddySilberhornRaymer2015,AnsariEtAl2018,SerinoEtAl2023}. The experimental implementation of the PVM construction described here is an essential building block of the Bell experiment for REVs presented in the main text.

Let the finite reference frame be represented by one optical mode
\( r\), prepared in a state \(\rho_\mathsf{ A}\) with finite coherent
amplitude. For example, for a coherent reference state
\(\rho_\mathsf{A}=\ket{\beta}\!\bra{\beta}\), one has
\(r\ket{\beta}=\beta\ket{\beta}\) with finite
\(\vert \beta\vert <\infty\). Its
quadratures are
\be
     X_\mathsf{A}
    =
    \frac{ r+ r^\dagger}{\sqrt2},
    \qquad
    P_\mathsf{A}
    =
    \frac{ r- r^\dagger}{i\sqrt2}.
\ee 
A general quadrature of the reference mode is then given by
\be
     X_\mathsf{A}( \vartheta)
    =
     X_\mathsf{A}\cos\vartheta
    +
     P_\mathsf{A}\sin\vartheta .
\ee
The \(N\) signal pulses are described by modes
\( a_1,\ldots, a_N\), with quadratures
\( X_{\mathsf a_i}\),  \(P_{\mathsf a_i}\) and \(X_\mathsf{a_i}(\vartheta)
    =
     X_\mathsf{a_i}\cos\vartheta
    +
     P_\mathsf{a_i}\sin\vartheta. \) The relational block
observables to be measured is
\begin{equation}
    X_\mathsf{rel}^{(N)}(\vartheta)
    =
    \frac1N\sum_{i=1}^N X_{\mathsf a_i} (\vartheta)
    -
    X_\mathsf{A}(\vartheta).
    \label{eq:pvm-optics-relational-block-observables}
\end{equation}

The first experimental step is to coherently combine the \(N\) signal pulses
into the normalized collective mode
\begin{equation}
    c   :=
    \frac1{\sqrt N}\sum_{i=1}^N \mathsf{a}_i .
    \label{eq:pvm-optics-collective-signal-mode}
\end{equation}
This can be implemented by a passive linear-optical unitary acting on the
\(N\) input modes and producing \(N\) output modes, chosen such that one output mode is
the equal superposition of all \(N\) input modes. In spatial modes this is an \(N\)-port
interferometer, for example decomposed into beamsplitters and phase shifters (see Fig.~\ref{fig:pvm-experimental-sequence-main}).
In a pulsed implementation it is the analogous temporal-mode unitary, which
maps the equally weighted temporal superposition of the \(N\) pulses onto one
output temporal mode. (The remaining \(N-1\) orthogonal modes are irrelevant
for the measurement.) Then
\be
    \frac1N\sum_{i=1}^N  X_{\mathsf a_i} (\vartheta)
    =
    \frac{X_c(\vartheta)}{\sqrt N}.
\ee

The collective signal mode \( c\) is then mixed with the finite reference
mode \( r\) on an unbalanced beamsplitter. We choose the beamsplitter
such that one output mode is
\begin{equation}
    d
    =
    \frac1{\sqrt{N+1}} c
    -
    \sqrt{\frac{N}{N+1}}\,r .
    \label{eq:pvm-optics-output-mode}
\end{equation}
This is a normalized optical mode. A balanced homodyne detector applied to the output mode \( d\) measures
the quadrature
\be
   X_d(\vartheta)
    =
    X_d\cos\vartheta+\hat P_d\sin\vartheta.
\ee
With \(s_N:=\sqrt{1+1/N}\) they satisfy
\be
    X_\mathsf{rel}^{(N)} (\vartheta) =s_N\hat X_d (\vartheta).
\ee

In the ideal homodyne limit, the detector realizes the quadrature PVM
\(\{\id_\Delta(\hat X_d(\vartheta))\}_\Delta\). Therefore,  for
every Borel set \(\Delta\subset\mathbb R\),  the optical circuit
realizes
\be
    \id_\Delta\!\left(\hat X_\mathsf{rel}^{(N)}(\vartheta)\right)
    =
    \id_{\Delta/s_N}\!\left(\hat X_d(\vartheta)\right), \mbox{ where } 
\Delta/s_N
    :=
    \{x\in\mathbb R\,:\,s_Nx\in\Delta\},
\ee
so the measured distribution is exactly the distribution required by the
PVM construction, up to the known rescaling \(s_N\). In the limit of infinitely large blocks, one has $s_N \rightarrow 1$.

The experimental sequence is therefore simple. For each block, prepare a
fresh finite reference pulse and \(N\) fresh signal pulses. Coherently map the
signal pulses to the collective mode \( c\), mix \( c\) with the finite
reference mode \( r\) on the unbalanced beamsplitter
\eqref{eq:pvm-optics-output-mode}, homodyne the output mode \( d\) and rescale the result. Repeating the procedure yields the block-outcome distribution of the
relational quadrature. 

The implementation of the Bell experiment based on the local measurements proposed here is given in the main text.

\end{document}

%% file: Preamble.tex
\usepackage{amsfonts,amsthm,amsmath, bbm}
\usepackage{nicefrac}
\usepackage{ragged2e}

\usepackage{dsfont}
\usepackage{verbatim}
\usepackage{stmaryrd}
\usepackage[normalem]{ulem}

\usepackage{wrapfig}
\usepackage{xcolor}
\usepackage{graphicx}

\definecolor{myred}{RGB}{220, 44, 61}
\definecolor{myyellow}{RGB}{230, 161, 0}
\definecolor{mygreen}{RGB}{0, 173, 107}
\definecolor{mygray}{RGB}{184, 184, 184}
\definecolor{mypurple}{RGB}{110, 35, 135}
\definecolor{color1}{RGB}{160,100,150}
\definecolor{color2}{RGB}{115,56,106}

\definecolor{myblue}{RGB}{12, 75, 190}
\definecolor{navyblue}{RGB}{20,60,120}
\definecolor{myorange}{RGB}{242,185,80}
\definecolor{myred}{RGB}{230,110,90}
\definecolor{mintgreen}{RGB}{90,220,210}
\definecolor{myplum}{RGB}{90,40,90}
\definecolor{terracota}{RGB}{200,110,80}
\definecolor{sienna}{RGB}{180,85,60}

\newtheorem{theorem}{Theorem}

\usepackage[colorlinks=true]{hyperref}

\hypersetup{
    linkcolor   = myblue,
    citecolor   = myred,
    urlcolor    = myplum,
}

\newcommand{\cG}{{\mathcal G}}
\newcommand{\cH}{{\mathcal H}}

\newcommand{\SU}{\mathrm{SU}}

\newcommand{\be}{\begin{equation}}
\newcommand{\ee}{\end{equation}}
\newcommand{\beq}{\begin{eqnarray}}
\newcommand{\eeq}{\end{eqnarray}}
\newcommand{\bes}{\begin{eqnarray}}
\newcommand{\ees}{\end{eqnarray}}

\newcommand{\q}{\quad}

\newcommand{\ecr}[1]{%
\noindent\textcolor{blue}{ #1}}

\newcommand{\sgn}{\mathrm{sgn}}
\newcommand{\tr}{{\mathrm{Tr}}}

\newcommand{\id}{\mathbb{I}}

\def\rd{\textrm{d}}

\def\be{\begin{equation}}
\def\bes{\begin{equation*}}
\def\ee{\end{equation}}
\def\ees{\end{equation*}}
\def\ba{\begin{aligned}}
\def\ea{\end{aligned}}
\def\bea{\be\ba}

\def\eea{\ea\ee}

\def\l{\left}
\def\r{\right}
\def\ls{\l\vert \l\vert }
\def\rs{\r\vert \r\vert }
\def\btj{\l(\begin{array}{ccc}}
\def\etj{\end{array}\r)}

\def\bsj{\l\{\begin{array}{ccc}}
\def\esj{\end{array}\r\}}

\def\bmat{\l(\begin{array}}
\def\emat{\end{array}\r)}

\newcommand{\ket}[1]{\left\vert  #1 \right\rangle}
\newcommand{\bra}[1]{\left\langle #1 \right\vert }

\newcommand{\ketbra}[2]{\left\vert  #1 \middle\rangle\!\middle\langle #2 \right\vert }

\makeatletter
\let\l@subsubsection\@gobbletwo
\makeatother